\documentclass[journal]{IEEEtran}
\ifCLASSINFOpdf
\else
\fi

\usepackage{graphicx}
\usepackage{subfigure}
\usepackage{cite}
\usepackage{algorithm}
\usepackage{algorithmic}
\usepackage{amssymb,amsmath,amsthm,amsfonts}
\usepackage{float}
\usepackage{subfigure}
\usepackage{multirow}
\usepackage{tabularx}
\usepackage{multirow}
\usepackage{makecell}
\usepackage{enumerate}
\usepackage{epstopdf}
\newtheorem{theorem}{Theorem}

\usepackage{array}
\begin{document}
%
\title{Exploring Auxiliary Context: Discrete Semantic Transfer  Hashing for Scalable Image Retrieval}
%
%
%

\author{Lei~Zhu,
        Zi~Huang,
        Zhihui Li,
        Liang Xie,
        Heng Tao Shen
\thanks{L. Zhu is with the School of Information Science and Engineering, Shandong Normal University, Jinan 250358, China.}
\thanks{Z. Huang are with the School of Information Technology and Electrical
Engineering, The University of Queensland, Brisbane, QLD 4072, Australia.}
\thanks{Z. Li is with  School of Computer Science and Technology, Shandong University, Jinan 250101, China.}
\thanks{L. Xie is with School of Sciences, Wuhan University of Technology, No.
122 Luoshi Road, Hongshan District, Wuhan 430070, China.}
\thanks{H. T. Shen (Corresponding author) is with the School of Computer Science and Engineering, University of Electronic
Science and Technology of China, Chengdu 611731, China (e-mail: shenhengtao@hotmail.com).}
}

%
%

\markboth{IEEE TRANSACTIONS ON NEURAL NETWORKS AND LEARNING SYSTEMS}%
{Shell \MakeLowercase{\textit{et al.}}: Bare Demo of IEEEtran.cls for IEEE Journals}
%



\maketitle

\begin{abstract}
Unsupervised hashing can desirably support scalable content-based image retrieval (SCBIR) for its appealing advantages of semantic label independence, memory and search efficiency. However, the learned hash codes are embedded with limited discriminative semantics due to the intrinsic limitation of image representation. To address the problem, in this paper, we propose a novel hashing approach, dubbed as \emph{Discrete Semantic Transfer Hashing} (DSTH). The key idea is to \emph{directly} augment the semantics of discrete image hash codes by exploring auxiliary contextual modalities. To this end, a unified hashing framework is formulated to simultaneously preserve visual similarities of images and perform semantic transfer from contextual modalities. Further, to guarantee direct semantic transfer and avoid information loss, we explicitly impose the discrete constraint, bit--uncorrelation constraint and bit-balance constraint on hash codes. A novel and effective discrete optimization method based on augmented Lagrangian multiplier is developed to iteratively solve the optimization problem. The whole learning process has linear computation complexity and desirable scalability. Experiments on three benchmark datasets demonstrate the superiority of DSTH compared with several state-of-the-art approaches.
\end{abstract}


%
\IEEEpeerreviewmaketitle

\section{Introduction}
\IEEEPARstart{W}{ith} the explosive growth in popularity of social networks and mobile devices, huge amounts of images are shared on the Web. There is an emerging need to retrieve relevant visual contents from such large-scale image databases with well scalability. Hence, scalable content-based image retrieval (SCBIR) has received substantial attentions over the past decades \cite{CBIRSurvey}.

\begin{figure}
\centering
\includegraphics[width=85mm]{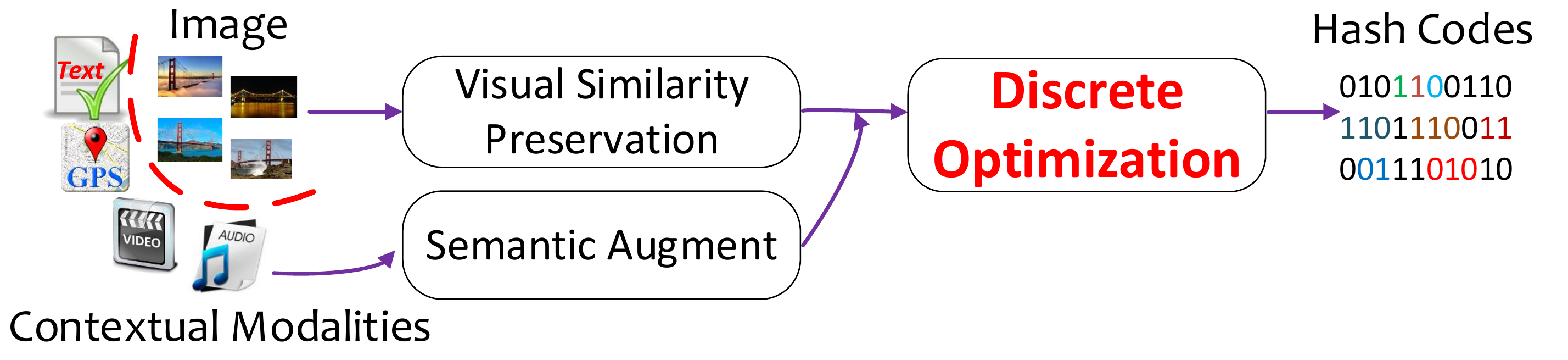}
\caption{Basic flowchart of hash code learning in DSTH.  Our approach \emph{directly} augments the semantics of discrete image hash codes with auxiliary contextual modalities.}
\label{fig:framework}
\end{figure}

Unsupervised hashing has been developed as one of the promising hashing techniques to support SCBIR \cite{SPH,SKLSH,DBLP:journals/sigpro/XieZPL16,AGH,CMFSIGMOD,ITQ,LCMH,zhutkde,zhutycb,SGH,TIP2016binary}. The key objective is to transform the high-dimensional image feature into compact binary codes with various advanced unsupervised learning techniques. By using binary codes as new representation, the memory consumption can be significantly reduced and the search process can be quickly completed with simple but efficient bit operations. Moreover, the learning process is performed without any dependence on semantic labels. Motivated by these desirable advantages, unsupervised hashing has recently received increasing attentions.

However, due to the intrinsic semantic limitation of image representation, the hash codes learned on it may suffer from limited discriminative representation capability \cite{zhuhashingmm}. How to enrich the semantics of image hash codes for SCBIR is an important but challenging task. Fortunately, the images to be retrieved by current Web search engines are generally accompanied with rich contextual modalities, such as text descriptions, GPS positions, audios, and etc \cite{zhuhashingmm,DBLP:journals/mta/XieZC16}. These resources of various modalities are noisy but easily to obtain. More importantly, they are semantically relevant to image data. It is promising to exploit them for semantic enrichment of image hash codes. Existing cross-modal hashing (CMH) (e.g. inter-media hashing (IMH) \cite{CMFSIGMOD} and linear cross-modal hashing (LCMH) \cite{LCMH}) can leverage contextual semantics. But their main objective is to discover the shared semantic space for cross-modal retrieval. Hence, the original visual information may be lost because of the mandatory heterogeneous modality correlation (validated in our experiments). Multi-modal hashing (MMH) (e.g. multiple feature hashing (MFH) \cite{MFH} and multi-view latent hashing (MVLH) \cite{MVLH}) can also enrich the semantics of hash codes. However, it requires both images and contextual modalities as query, which impedes its application for SCBIR where only visual image is provided for online retrieval.
\begin{table*}
\caption{Main differences between DSTH and representative hashing techniques. CM denotes the contextual modalities.}
\label{difference}
\centering
\begin{tabular}{|p{10mm}<{\centering}|p{19mm}<{\centering}|p{22mm}<{\centering}|p{28mm}<{\centering}|p{27mm}<{\centering}|p{32mm}<{\centering}|p{10mm}<{\centering}|}
\hline
Method & Query Modality & Learning Modality & Learning Paradigm  & Semantic Transfer & Discrete Optimization & SCBIR\\
\hline
SGH & visual & visual & unsupervised & $\times$ & $\times$ & $\surd$\\
\hline
DGH & visual & visual & unsupervised & $\times$ & $\surd$ & $\surd$  \\
\hline
SDH & visual & visual & supervised & $\surd$ & $\surd$ & $\surd$\\
\hline
CMFH & visual/CM & visual+CM & unsupervised & $\times$ & $\times$ & $\surd$ \\
\hline
CDH &  visual/CM & visual+CM & supervised & $\surd$ & $\surd$ & $\surd$ \\
\hline
MFH & visual+CM & visual+CM & unsupervised & $\surd$ & $\times$ & $\times$ \\
\hline
DSTH & visual & visual+CM & unsupervised & $\surd$ & $\surd$ & $\surd$ \\
\hline
\end{tabular}
\end{table*}

In this paper, we propose a novel hashing method, dubbed as \emph{Discrete Semantic Transfer Hashing} (DSTH). The key idea is to directly augment the semantics of discrete hash codes with auxiliary contextual modalities. To achieve this nontrivial aim, DSTH first aligns image hash codes with topic distributions of contextual modalities for semantic transfer. Then, a unified hashing learning framework is formulated to integrate semantic transfer with visual similarity preservation. These two parts interact with each other and guarantee that the valuable semantics can be transferred to image hash codes. Further, DSTH simultaneously imposes discrete constraint, bit--uncorrelation constraint, and bit-balance constraint on hash codes. It can avoid the semantic loss brought by most existing hashing methods which follow a two-step relaxing+rounding optimization framework. An efficient and effective optimization method based on augmented Lagrangian multiplier (ALM) \cite{ALM} is proposed to iteratively solve the discrete hash codes. The whole learning process has linear computation complexity and desirable scalability. Figure \ref{fig:framework} illustrates the basic process of hash code learning in DSTH. It is worthwhile to highlight the main contributions of this paper as follows:
\begin{enumerate}
 \item DSTH exploits the auxiliary contextual modalities to directly augment the semantics of discrete image hash codes. It can support image retrieval where only visual query is provided. To the best of our knowledge, there does not exist any similar work. \vspace{1mm}
 \item To ensure direct semantic transfer and avoid information loss, DSTH explicitly deals with discrete constraint, bit--uncorrelation constraint, and bit-balance constraint together. A novel and efficient optimization approach based on augmented Lagrangian multiplier is developed to directly learn discrete hash codes. The learning process has linear computation complexity and desirable scalability.\vspace{1mm}
 \item Extensive experiments demonstrate the state-of-the-art performance of DSTH, and also validate the effects of semantic transfer and discrete optimization.\vspace{1mm}
\end{enumerate}

The rest of the paper is structured as follows. Section \ref{sec:2} reviews the related work.
Details about the proposed methodology are presented in Section \ref{sec:3}. In Section \ref{sec:4}, we introduce the experiments.
Section \ref{sec:5} concludes the paper.

\section{Related Work}
\label{sec:2}
Hashing is a quite hot research topic in recent literatures on image indexing. Various approaches are developed in this research field. For the limited space here, only the most related works of this paper are reviewed in this section. For more comprehensive introduction, please refer to \cite{hashingsurveytwo}.

\subsection{Data-independent Hashing}
Locality-sensitive hashing (LSH) \cite{SKLSH} and its extensions are typical data-independent hashing methods. They generate binary codes via random projection. For example, the hash functions of LSH are constructed with the random vectors from a standard Gaussian distribution. As their whole hash code generation process is performed without considering any semantics of underlying image data, data-independent hashing methods generally require more hashing bits and tables to achieve a satisfactory performance. It will result in longer search time and significant storage cost. To enrich the hash codes with semantics, advanced machine learning techniques are applied for hashing. With the trend, various data-dependent hashing methods (supervised and unsupervised hashing) are proposed to capture the data characteristics and embed them into binary hash codes.

\subsection{Supervised Hashing}
Supervised hashing learns hash codes with explicit semantic labels. Via supervised learning, discriminative capability of hash codes can be enhanced by mining semantics in explicit labels. Typical examples include kernel-based supervised hashing (KSH) \cite{SKH}, semantic correlation maximization (SCM) \cite{SCM}, semantics-preserving hashing (SePH) \cite{SePH}, linear subspace ranking hashing (LSRH) \cite{LSRH}, and deep learning hashing (DLH) \cite{DHCB,shen2018tpami}. In KSH, Hamming distances between hash codes of similar data pairs are minimized and that of dissimilar data pairs are maximized simultaneously. SCM solves the training time complexity of supervised multimodal hashing methods by avoiding explicit pairwise similarity matrix computing. SePH transforms the supervised semantic affinities of training data into a probability distribution and approximates it with hash codes in Hamming space. LSRH is a typical ranking-based cross-modal hashing with supervised learning. It considers a new class of hash functions that are closely related to rank correlation measures. DLH directly projects original images to binary hash codes via multiple hierarchical nonlinear transformation in a deep neural network \cite{Alexnet}. Principally, supervised hashing can indeed achieve better performance than unsupervised hashing. However, they require large amounts of high-quality semantic labels to achieve satisfactory performance. This requirement unfortunately limits the retrieval scalability of hashing in practical image retrieval, where high-quality semantic labels are hard and expensive to obtain.

\subsection{Unsupervised Hashing}
Unsupervised hashing generates hash codes without any semantic labels. It has a better scalability. According to the exploited modalities, it can be further categorized into three sub-categories: unsupervised uni-modal hashing, unsupervised cross-modal hashing, and unsupervised multi-modal hashing.

\textbf{Unsupervised Uni-modal Hashing}. Its learning process only relies on discriminative information in single modality. For image retrieval, only visual information is considered. Typical examples include: spectral hashing (SPH) \cite{SPH}, anchor graph hashing (AGH) \cite{AGH}, iterative quantization (ITQ) \cite{ITQ}, scalable graph hashing (SGH) \cite{SGH}, discrete proximal linearized minimization (DPLM) \cite{TIP2016binary}, latent semantic minimal hashing (LSMH) \cite{LSMH}, and Deepbit \cite{deepbit}. SPH preserves the image similarities into the projected hash codes with spectral graph. To reduce the training complexity of graph hashing, AGH approximates the image relations with a low-rank matrix, based on which hash functions are learned by binarizing the Nystorm eigen-functions \cite{SPH}. ITQ minimizes the quantization loss brought by dimension reduction based binary embedding. SGH applies feature transformation to solve large-scale graph hashing. DPLM reformulates the unsupervised discrete hashing learning problem as minimizing the sum of a smooth loss term. The transformed problem can be efficiently solved with an iterative procedure where each iteration admits an analytical discrete solution. LSMH combines minimum encoding and matrix decomposition to learn the hash codes based on the refined feature representation. Deepbit is one of the pioneering unsupervised deep hashing methods, where three criterions are enforced to learn hash codes from the top layer of the designed neural network. As the learning process is independent with semantic labels, unsupervised uni-modal hashing has well scalability when they process large-scale data. However, image features have the intrinsic limitation on representing high-level semantics. Hence, hash codes learned by this kind of method will inevitably suffer from semantic shortage. This disadvantage limits the retrieval performance of unsupervised uni-modal hashing.

\textbf{Unsupervised Cross-modal Hashing}. This method can exploit contextual modalities to learn hash code supporting the retrieval tasks across different modalities \cite{DBLP:conf/aaai/XieSZ16}. Cross-view hashing (CVH) \cite{CMFIJCAI} extends spectral hashing into cross-modal retrieval by minimizing similarity-weighed Hamming distance of the learned codes. Inter-media hashing (IMH) \cite{CMFSIGMOD} simultaneously preserves inter-modality similarities and correlates heterogeneous modalities to learn cross-modal hash codes. Linear cross-modal hashing (LCMH) \cite{LCMH} further reduces training complexity of IMH by representing training samples as their distances to centroids of sample clusters. Collective matrix factorization hashing (CMFH) \cite{TIPCMFH} seeks to detect the shared latent structures of heterogeneous modalities with collective matrix factorization \cite{Singh:2008:RLV:1401890.1401969} for generating cross-modal hash codes. These cross-modal hashing methods can enhance the descriptive capability of image hash codes with shared semantic learning. However, the main objective of cross-modal hashing is discovering the shared hash codes in heterogeneous modalities to achieve cross-modal retrieval. The valuable semantics in original visual features may unfortunately lost as a result of mandatory heterogeneous modality correlation. For specific task of SCBIR, the lost visual information may deteriorate the image search performance.

\textbf{Unsupervised Multi-modal Hashing}. Motivated by the success of multiple feature fusion on enhancing the performance \cite{DBLP:journals/tcyb/ZhuSJZX15,DBLP:journals/tmm/ZhuSJXZ15,DBLP:journals/tcyb/ChangMYZH17,DBLP:journals/tnn/ChangY17,DBLP:journals/pami/ChangYYX17}, this method integrates semantics of contextual modalities and visual information into a unified hash code. Composite hashing with multiple information source (CHMIS) \cite{CHMIS} is one of the pioneering approaches. It simultaneously adjusts the weights of modalities to maximize the coding performance, and learns hash functions for fast query binary transformation. Multi-view spectral hashing (MVSH) \cite{MVSH} extends spectral hashing into a multi-view setting. Its key idea is to sequentially learn the integrated hash codes by solving the successive maximization of local variances. In \cite{6638233}, an efficient multi-view anchor graph hashing (MVAGH) is proposed to learn the nonlinear integrated hash codes computed from the eigenvectors of an averaged similarity matrix. The training complexity of multi-view hashing is reduced with a low-rank from of the averaged similarity matrix calculated based on multi-view anchor graphs. Multi-view alignment hashing (MVAH) \cite{7006770} learns the relaxed hashing representation with a regularized kernel nonnegative matrix factorization, and hash functions via multivariable logistic regression. Multi-view latent hashing (MVLH) \cite{MVLH} incorporates multi-modal data into hash code learning. In MVLH, the hash codes are determined as the latent factors shared by multiple views from an unified kernel feature space. Compared with uni-modal hashing methods, multi-modal hashing can generate more discriminative hash codes. However, it needs all the involved modalities as input \cite{zhuijcai,7984879}. This requirement cannot be satisfied as only visual image is provided in SCBIR.

\subsection{Discrete Hashing}
Most existing hashing approaches exploit a two-step relaxing+rounding to solve hash codes. In these methods, the relaxed hash codes (continuous values) are first learned and further quantized into binary codes via thresholding. This solution, as indicated in recent literature \cite{DGH,SDH}, may lead to significant information loss. To address the problem, several approaches are proposed to directly solve hash codes within one step. Discrete graph hashing (DGH) \cite{DGH} aims to directly preserve the data similarity in a discrete Hamming space. It reformulates the unsupervised graph hashing with a discrete optimization framework and solves two subproblems via a tractable alternating maximization. Supervised discrete hashing (SDH) \cite{SDH} learns the supervised discrete hash codes with the optimal linear classification performance. SDH transforms this learning objective into sub-problems that can admit an analytical solution. A cyclic coordinate descent algorithm is applied to calculate discrete hashing bits in a closed form. Coordinate discrete hashing (CDH) \cite{CDOECVR} is designed for supervised cross-modal hashing, and its discrete optimization proceeds in a block coordinate descent manner. In each iterative learning step, a hash bit is sequentially updated while others clamped. CDH transforms the sub-problem into an equivalent and tractable
quadratic form, so that hash codes can be directly solved with active set based optimization. Column sampling based discrete supervised hashing (COSDISH) \cite{CSBDSH} and kernel-based supervised discrete hashing (KSDH) \cite{Shi2016} are also developed for supervised hashing. COSDISH operates in an iterative manner. In each iteration, several columns are first sampled from the semantic similarity matrix. Then, hash code is decomposed into two parts, so that it can be alternately optimized. KSDH solves discrete hash codes via an asymmetric relaxation strategy that preserves the discrete constraint and reduces the accumulated quantization errors. Although these approaches can achieve certain success, their proposed discrete optimization solutions are specially designed for particular hashing types (supervised hashing, unsupervised graph hashing, and etc.). Therefore, they cannot be directly applied to handle our problem.

\subsection{Key Differences between Our Approach and Existing Works}
Our work is an advocate of discrete hashing optimization but focuses on the problem of exploiting contextual modalities to directly augment the semantics of discrete hash codes. Moreover, our hashing optimization strategy can not only explicitly deal with the discrete constraint of binary codes, but also consider the bit--uncorrelation constraint and bit-balance constraint together\footnote{DPLM can cope with bit-uncorrelation and bit-balance constraints. However, it simply transfers two constraints to the objective function and avoids to directly solve the problem.}. The whole learning
process has linear computation complexity and desirable scalability. The proposed approach can well support SCBIR. The main differences between DSTH and existing hashing techniques are summarized in Table \ref{difference}.

\begin{table}
\caption{Summary of main notations.}
\vspace{-3mm}
\label{symtable}
\centering
\begin{tabular}{|p{12mm}<{\centering}|p{67mm}|}
\hline
\textbf{Symbols} & \textbf{Explanations} \\
\hline
\hline
$X$ & feature representations of images \\
\hline
$Y$ & feature representations of contextual texts\\
\hline
$Z$ & hash codes of images\\
\hline
$H$ & projection matrix in hash functions\\
\hline
$d_x$ & feature dimension of image representation \\
\hline
$d_y$ & feature dimension of contextual text representation \\
\hline
$N$ & number of database images \\
\hline
$F$ & group of hash functions \\
\hline
$K$ & number of anchors \\
\hline
$L$ & hash code length \\
\hline
$S_x$ & affinity matrix of visual graph \\
\hline
$L_x$ & Laplacian matrix of visual graph \\
\hline
$V_x$ & data-to-anchor mapping matrix\\
\hline
$U$ & basis matrix of visual feature space\\
\hline
$T$ & latent image semantic topics\\
\hline
$W$ & semantic transfer matrix\\
\hline
$A_x, A_y, B$ & auxiliary variables\\
\hline
$E_x,E_y,E_z$ & measure the difference between the target and auxiliary variables\\
\hline
\end{tabular}
\end{table}

\section{The Proposed Methodology}
\label{sec:3}
In this subsection, we detail the proposed methodology. First, we introduce the relevant notations used in this paper and the problem definition. Then, we formulate the overall objective function and present an efficient discrete solution. Finally, we analyse the convergence and time complexity of the proposed iterative optimization method.

\subsection{Notations and Problem Definition}
In this paper, we explore to exploit semantics from contextual texts for semantic transfer of discrete hash codes. Note that our approach can be easily extended when more contextual modalities are exploited. We define $X=[x_1,...,x_N]\in \mathbb{R}^{d_x\times N}$ and $Y=[y_1,...,y_N]\in \mathbb{R}^{d_y\times N}$ as feature representations of images and the contextual texts respectively, $d_x$ and $d_y$ denote their corresponding feature dimensions, and $N$ is the number of images. The objective of DSTH is to learn $Z=[z_1, z_2,...,z_N]\in \mathbb{R}^{L\times N}$, where $z_n=[z_{1n}, z_{2n},$ $...,z_{Ln}]^{\texttt{T}}\in \mathbb{R}^{L\times 1}$ are the hash codes of the $n_{th}$ image, $L$ is hash code length. To generate hash codes for query images, DSTH learns a group of hash functions $F$, each of them defines a mapping: $\mathbb{R}^{d_x}\mapsto \{0,1\}$.  Main notations used in the paper are listed in Table \ref{symtable}.

\subsection{Objective Formulation}
The formulated objective is composed of two parts: visual similarity preservation and semantic transfer. Visual similarity preservation preserves visual correlation of images into hash codes. Semantic transfer part discovers the potential semantics from contextual texts and transfers them into discrete hash codes.

\textbf{Visual Similarity Preservation}. SCBIR retrieves similar images for the query \cite{DBLP:conf/mm/ZhuSX15}. Hence, the objective of hashing for SCBIR is visual similarity preservation. It indicates that similar images should be mapped to binary codes with short Hamming distances.
In this paper, we seek to minimize the weighted Hamming distance of hash codes.
\begin{equation}
\begin{aligned}
\label{eq:sh}
\min_{\{z_i\}_{i=1}^N} \sum_{i=1}^N \sum_{j=1}^N S_x(i,j) ||z_i-z_j||_F^2 \Rightarrow \min_{Z} \ Tr(ZL_xZ^{\texttt{T}})\\
\end{aligned}
\end{equation}

\noindent where $L_x=D_x-S_x$ is the Laplacian matrix of visual graph, $S_x\in \mathbb{R}^{N\times N}$ characterizes the affinity similarities of images, $D_x=S_x\textbf{1}=I$, $\textbf{1}$ is column vector with ones, and $Tr(\cdot)$ is trace operator,  $||\cdot||_F$ is Frobenius norm. The design principle of the Eq.(\ref{eq:sh}) is to incur a heavy penalty if two similar images are projected far apart.

Explicitly computing $L_x$ will consume $O(N^2)$, which is not scalable for large-scale image retrieval. In this paper, we exploit anchors to reduce the computation complexity. Similar to \cite{AGH}, we approximate the affinity matrix $S_x$ with $S_x=V_x\Lambda V_x^\texttt{T}$, where $\Lambda=\texttt{diag}(V_x^\texttt{T}\textbf{1})$. $V_x=[v(x_1),..,v(x_{N})]^\texttt{T}$, $v(x)$ is data-to-anchor mapping
\begin{equation*}
\begin{aligned}
v(x)=\frac{[\delta_1\texttt{exp}(\frac{-||x-r_1||_2^2)}{\sigma}), ..., \delta_K \texttt{exp}(\frac{-||x-r_K||^2_2}{\sigma})]^\texttt{T}}{{\sum_{k=1}^K \delta_k \texttt{exp}(\frac{-||x-r_k||_2^2}{\sigma})}}
\end{aligned}
\end{equation*}
\noindent $r_1,...,r_K$ are $K$ anchors obtained by \emph{k}-means, $\delta_k$ is set to 1 if $r_k$ belongs to the $s$ closest exemplars of $x$, and 0 vice versa, $\sigma>0$ is the bandwidth parameter. Accordingly, $L_x$ can be represented as $I-V_x\Lambda V_x^\texttt{T}$. As shown in the subsequent discrete hashing learning, keeping this low-rank form decomposition will avoid explicit Laplacian matrix computation, and reduce the computation complexity of optimization.

\textbf{Semantic Transfer}. Images in the modern searching engines are generally associated with rich textual descriptions, such as tags, image captions and user comments. These images and the accompanied texts belong to heterogeneous modalities but may be highly correlated with each other. Moreover, contextual texts contain explicit semantics which are complementary to the latent image semantics. Hence, it is promising to exploit contextual modalities for the semantic enrichment of discrete image hash codes. To this end, in this paper, we first adopt matrix factorization to detect the latent semantic structure $T$ of image. Its formulation is $\min_{U, T} \ ||X-UT||_F^2$,
where $U\in \mathbb{R}^{d_x\times L}$ is basis matrix of visual feature space, $T\in \mathbb{R}^{L\times N}$ represents latent image semantic topics. Then, for semantic transfer, we align latent image semantics $T$ to explicit textual semantic distribution $Y$
\begin{equation}
\begin{aligned}
\label{eq:td}
\min_{W, T} \ & ||WT-Y||_F^2\\
\end{aligned}
\end{equation}
\noindent where $W\in \mathbb{R}^{d_y\times L}$ is the transfer matrix. With transfer, the detected $T$ can involve the explicit semantics of contextual text. In hashing learning, we directly force hash codes $Z$ to match the distribution of $T$. This design is reasonable because the hash codes can be understood as semantic topic distribution, if we consider each hashing bit as a latent semantic topic.

\textbf{Imposing Constraints}. In our formulation, we explicitly consider three constraints on hash codes to ensure direct semantic transfer and avoid information quantization loss. $Z\in \{-1, 1\}^{L\times N}$ is discrete constraint. It guarantees any hash code to be $-1$ or 1. Via simple transformation $(Z+1)/2$, $Z$ will be binary code (0 or 1). With binary codes as image representation, the search process can be significantly accelerated and the storage cost of image database can be greatly reduced. The bit-uncorrelation constraint $ZZ^\texttt{T}=NI$ is to guarantee the learned hashing bits to be uncorrelated. It can reduce the information redundancy of different hash bits. $Z\textbf{1}=0$ is the bit-balance constraint, it requires each bit to occur in database with equal chance ($50\%$). This constraint forces the learned hash code to contain the largest information.

\textbf{Overall Formulation}. After comprehensively considering visual similarity preservation, semantic transfer and constraints to be imposed, we obtain the overall objective function of DSTH. The formulation is
\begin{equation}
\begin{aligned}
\label{eq:objective}
\min_{Z, U, W} \ & ||X-UZ||_F^2+\beta ||WZ-Y||_F^2 + \alpha Tr(Z(I-V_x\Lambda V_x^\texttt{T})Z^{\texttt{T}})\\
s.t. \ & \underbrace{Z\in \{-1, 1\}^{L\times N}}_{discrete}, \underbrace{ZZ^\texttt{T}=NI}_{bit-uncorrelation}, \underbrace{Z\textbf{1}=0}_{bit-balance}
\end{aligned}
\end{equation}
\noindent where $\alpha, \beta>0$ balances the regularization terms. We jointly consider visual similarity preservation and semantic transfer, so that visual similarity preservation can guide the semantic extraction and determine which part of semantics to transfer.


\subsection{Discrete Optimization}
Solving Eq.(\ref{eq:objective}) is essentially a non-trivial combinatorial optimization problem for three challenging constraints. Most existing hashing approaches apply relaxing+rounding optimization \cite{hashingsurveytwo}. They first relax discrete constraint to calculate continuous values, and then binarize them to hash codes via rounding. This two-step learning can simplify the solving process, but it may cause significant information loss. In recent literature, several discrete hashing solutions are proposed. However, they are developed for particular hashing types and formulations. For example, graph hashing \cite{DGH}, supervised hashing \cite{SDH,CSBDSH}, cross-modal hashing \cite{CDOECVR}. Therefore, their learning approaches cannot be directly applied to solve our problem.

In this paper, we propose a new and effective optimization algorithm based on augmented Lagrangian multiplier (ALM) \cite{ALM}.
Our idea is to introduce auxiliary variables to separate constraints, and transform the objective function to an equivalent one that is tractable.
Formally, we introduce three auxiliary variables
$A_x, A_y, B$, and set $A_x=X-UZ, A_y=Y-WZ, B=Z$. Eq.(\ref{eq:objective}) is reformulated as
\begin{equation}
\begin{aligned}
\label{eq:tj}
\min_{Z,U,W} & ||A_x||_F^2+||A_y||_F^2+ \frac{\mu}{2}(||X-UZ -A_x+\frac{E_x}{\mu}||_F^2  + \\
&\beta ||Y-WZ-A_y+\frac{E_y}{\mu}||_F^2) + \alpha Tr(Z(I-V_x\Lambda V_x^\texttt{T})B^{\texttt{T}}) \\
&+ \frac{\mu}{2} ||Z-B+\frac{E_z}{\mu}||_F^2 \\
& s.t. \quad B\in \{-1, 1\}^{L\times N}, ZZ^\texttt{T}=NI, Z\textbf{1}=0
\end{aligned}
\end{equation}
\noindent where $E_x\in \mathbb{R}^{d_x\times N}$, $E_y\in \mathbb{R}^{d_y\times N}$, $E_z\in \mathbb{R}^{L\times N}$ measure the difference between the target and auxiliary variables, $\mu>0$ adjusts the balance between terms.

We can adopt alternate optimization to iteratively solve Eq.(\ref{eq:tj}). Specifically, we optimize the objective function with respective to one variable while fixing the other remaining variables. The iteration steps are detailed as follows.

\textbf{Update $A_x, A_y$.} By fixing other variables, the optimization formulas for $A_x, A_y$ are
\begin{equation}
\begin{aligned}
\label{eq:tk}
&\min_{A_x} \ ||A_x||_F^2+ \frac{\mu}{2}||X-UZ-A_x + \frac{E_x}{\mu}||_F^2 \\
&\min_{A_y} \ ||A_y||_F^2+ \frac{\mu}{2}||Y-WZ-A_y + \frac{E_y}{\mu}||_F^2 \\
\end{aligned}
\end{equation}
\noindent By calculating the derivative of the objective function with respective to $A_x, A_y$, and setting it to 0, we can obtain that
\begin{equation}
\begin{aligned}
\label{eq:tl}
A_x=\frac{\mu X-\mu UZ+E_x}{2+\mu}, A_y=\frac{\mu Y-\mu WZ+E_y}{2+\mu}
\end{aligned}
\end{equation}

\textbf{Update $U,W$}. By fixing other variables, the optimization formula for $U$ is
\begin{equation}
\begin{aligned}
\label{eq:tn}
\min_{U} \ & ||X-UZ-A_x + \frac{E_x}{\mu}||_F^2 \\
\end{aligned}
\end{equation}
\noindent By calculating the derivative of the objective function with respective to $U$, and setting it to 0, we can obtain that
\begin{equation}
\begin{aligned}
UZ = X-A_x+\frac{E_x}{\mu} \\
\end{aligned}
\end{equation}

\noindent Since $ZZ^\texttt{T}=NI$, we can further derive that
\begin{equation}
\begin{aligned}
\label{eq:tp}
U = \frac{1}{N}(X-A_x+\frac{E_x}{\mu})Z^\texttt{T}
\end{aligned}
\end{equation}
Similarly, we can obtain $W = \frac{1}{N}(Y-A_y+\frac{E_y}{\mu})Z^\texttt{T}$.

\textbf{Update $B$}. By fixing other variables, the optimization formula for $B$ is
\begin{equation}
\begin{aligned}
\label{eq:tq}
\min_{B} \ & \alpha Tr(Z(I-V_x\Lambda V_x^\texttt{T})B^\texttt{T})+\frac{\mu}{2}||Z-B+\frac{E_z}{\mu}||_F^2\\
s.t. \ & B\in \{-1, 1\}^{L\times N}\\
\end{aligned}
\end{equation}
\noindent The objective function in Eq.(\ref{eq:tq}) can be simplified as
\begin{equation}
\begin{aligned}
\label{eq:tre}
\min_{B} &  \ Tr((\frac{\alpha}{\mu}Z(I-V_x\Lambda V_x^\texttt{T})-Z-\frac{E_z}{\mu})B^\texttt{T})\\
= \min_{B} & \ ||B-(Z+\frac{E_z}{\mu}-\frac{\alpha}{\mu}Z(I-V_x\Lambda V_x^\texttt{T}))||_F^2\\
\end{aligned}
\end{equation}
The discrete solution of $B$ can be directly represented as
\begin{equation}
\begin{aligned}
\label{eq:tz}
B=\texttt{Sgn}(Z+\frac{E_z}{\mu}-\frac{\alpha}{\mu}ZI + \frac{\alpha}{\mu}ZV_x\Lambda V_x^\texttt{T})
\end{aligned}
\end{equation}
\noindent where $\texttt{Sgn}(\cdot)$ is signum function which returns -1 if $x<0$, 1 if $x\geq 0$.

\textbf{Update $Z$}. By fixing other variables, the optimization formula for $Z$ is
\begin{equation}
\begin{aligned}
\label{eq:tr}
& \min_{Z} \  \frac{\mu}{2}(||X-UZ-A_x+\frac{E_x}{\mu}||_F^2 + \beta ||Y-WZ-A_y \\
& +\frac{E_y}{\mu}||_F^2)+\alpha Tr(Z(I-V_x\Lambda V_x^\texttt{T})B^{\texttt{T}})+\frac{\mu}{2} ||Z-B+\frac{E_z}{\mu}||_F^2 \\
& s.t. \  ZZ^\texttt{T}=NI, Z\textbf{1}=0
\end{aligned}
\end{equation}
\noindent The objective function in Eq.(\ref{eq:tr}) can be transformed as
\begin{equation}
\begin{aligned}
\label{eq:trrr}
& \min_{Z} \   -\mu Tr(Z^\texttt{T}U^\texttt{T}(X-A_x+\frac{E_x}{\mu}))- \mu\beta Tr(Z^\texttt{T}W^\texttt{T}(Y-A_y \\
& + \frac{E_y}{\mu})) + \alpha Tr(Z^\texttt{T}B(I-V_x\Lambda V_x^\texttt{T}))-\mu Tr(Z^\texttt{T}(B-\frac{E_z}{\mu}))\\
& = \min_{Z} \ -Tr(Z^\texttt{T}C)
\end{aligned}
\end{equation}

\noindent where $C=B-\frac{E_z}{\mu}-\frac{\alpha}{\mu}BI+\frac{\alpha}{\mu}BV_x\Lambda V_x^\texttt{T}+U^\texttt{T}(X-A_x+
\frac{E_x}{\mu})+\beta W^\texttt{T}(Y-A_y+\frac{E_y}{\mu})$.
Eq.(\ref{eq:tr}) is equivalent to the following maximization problem
\begin{equation}
\begin{aligned}
\label{eq:ts}
& \max_{Z} \ Tr(Z^\texttt{T}C) \\
 s.t. & \ ZZ^\texttt{T}=NI, Z\textbf{1}=0
\end{aligned}
\end{equation}
By mathematically solving the above equation with singular value decomposition (SVD) \cite{wall2003svd}, $C$ can be decomposed as $C=P\Theta Q^\texttt{T}$,
where the columns of $P$ and $Q$ are left-singular vectors and right-singular vectors of $C$ respectively, $\Theta$ is rectangular diagonal matrix and its diagonal entries are singular values of $C$. Then, the optimizing for $Z$ becomes $\max_{Z} \ Tr(Z^\texttt{T}P\Theta Q^\texttt{T}) \Leftrightarrow \max_{Z} \ Tr(\Theta Q^\texttt{T}Z^\texttt{T}P)$.
\begin{theorem}\label{calzthre}
Given a matrix $G$ which meets $GG^\texttt{T}=NI$ and diagonal matrix $\Theta \ge 0$, the solution of $\max_{G} Tr(\Theta G)$ is $\texttt{diag}(\sqrt{N})$.
\end{theorem}
\begin{proof}
Let us assume $\theta_{ii}$ and $g_{ii}$ are the $i_{th}$ diagonal entry of $\Theta$ and $G$ respectively, $Tr(\Theta G)=\sum_{i}\theta_{ii}g_{ii}$. Since $GG^\texttt{T}=NI$, $g_{ii}\leq \sqrt{N}$. $Tr(\Theta G)=\sum_{i}\theta_{ii}g_{ii}\leq \sqrt{N}\sum_{i}\theta_{ii}$. The equality holds only when $g_{ii}=\sqrt{N}, g_{ij}=0, \forall i,j$. $Tr(\Theta G)$ achieves its maximum when $G=\texttt{diag}(\sqrt{N})$.
\end{proof}

$\Theta \ge 0$ as $\Theta$ is calculated by SVD. On the other hand, we can easily derive that $Q^\texttt{T}Z^\texttt{T}PP^\texttt{T}ZQ=NI$. Therefore, according to the \textbf{Theorem \ref{calzthre}}, the optimal $Z$ can only be obtained when $Q^\texttt{T}Z^\texttt{T}P=\texttt{diag}(\sqrt{N})$. Hence, the solution of $Z$ is
\begin{equation}
\begin{aligned}
\label{eq:solutiony}
Z=\sqrt{N}PQ^\texttt{T}
\end{aligned}
\end{equation}
Moreover, in order to satisfy the bit-balance constraint $Z\textbf{1}=0$, we apply Gram-Schmidt process as \cite{DGH} and construct matrices $\hat{P}$ and $\hat{Q}$, so that $\hat{P}^\texttt{T}\hat{P}=I_{L-R}$, $[P, 1]^\texttt{T}\hat{P}=0$, $\hat{Q}^\texttt{T}\hat{Q}=I_{L-R}$, $Q\hat{Q}^\texttt{T}=0$, $R$ is the rank of $C$. The close form solution for
$Z$ is
\begin{equation}
\begin{aligned}
Z=\sqrt{N}[P, \hat{P}][Q, \hat{Q}]^\texttt{T}
\end{aligned}
\end{equation}

\textbf{Update $E_x$, $E_y$, $E_z$, $\mu$}. The update rules are ($\rho>1$ is learning rate that controls the convergence.)
\begin{equation}
\begin{aligned}
\label{eq:tu}
&E_x=E_x+\mu(X-UZ-A_x)\\
&E_y=E_y+\mu(Y-WZ-A_y)\\
&E_z=E_z+\mu(Z-B), \mu=\rho\mu \\
\end{aligned}
\end{equation}

\textbf{Convergence}. At each iteration, the updating of variables will monotonically decreases towards the lower bound of objective function in Eq.(\ref{eq:objective}) . As indicated by ALM optimization theory \cite{lin2010augmented}, the iterations will make the optimization converge. Further, our empirical experiments on standard benchmarks also validate the convergence  of the proposed method.

\textbf{Hash Function Learning}. In this paper, we leverage linear projection to construct hash functions for its high online efficiency. The objective is to minimize the loss between the hash codes and the projected ones. The formulation is $\min_{H} ||Z-H^\texttt{T}X||_F^2 + \eta ||H||_F$, where $H\in \mathbb{R}^{d_x\times L}$ denotes the projection matrix. The optimal $H$ can be calculated as $H=\left(XX^\texttt{T}+\eta I\right)^{-1} XZ^\texttt{T}$. Hash functions can be constructed as $F(x)=\frac{\texttt{sgn}(H^\texttt{T}x)+1}{2}$. It should be noted that, since DSTH is a two-stage hashing framework, this hash function learning part can be substituted by other models, such as, linear SVM \cite{STH}, kernel logistic regression \cite{SePH}, decision tree \cite{DecisionTree}, and neural network \cite{7006770}.

\begin{table*}
\caption{mAP of all approaches on three datasets. The best result in each column is marked with bold.}
\label{resulttable}
\centering
\begin{tabular}{|p{16mm}<{\centering}|p{8mm}<{\centering}|p{8mm}<{\centering}|p{8mm}<{\centering}|p{11mm}<{\centering}
|p{8mm}<{\centering}|p{8mm}<{\centering}|p{8mm}<{\centering}|p{11mm}<{\centering}|p{8mm}<{\centering}|p{8mm}<{\centering}
|p{8mm}<{\centering}|p{11mm}<{\centering}|}
\hline
\multirow{2}{*}{Methods} & \multicolumn{4}{c|}{\emph{\textbf{Wiki}}} & \multicolumn{4}{c|}{\emph{\textbf{MIR Flickr}}} & \multicolumn{4}{c|}{\emph{\textbf{NUS-WIDE}}}\\
\cline{2-13}
& 16 & 32 & 64  & 128  & 16 & 32 & 64  & 128 & 16 & 32 & 64 & 128 \\
\hline
SKLSH & 0.1440 & 0.1509 & 0.1562 & 0.1609 & 0.5624 & 0.5902 & 0.6023 & 0.6253 & 0.3685 & 0.3645 & 0.3684 & 0.3688 \\
\hline
SPH & 0.1625 & 0.1583 & 0.1628 & 0.1714 & 0.5976 & 0.6093 & 0.6269 & 0.6433 & 0.3430 & 0.3910 & 0.4415 & 0.4582 \\
\hline
ITQ & 0.1717 & 0.1780 & 0.1789 & 0.1810 & 0.6194 & 0.6339 & 0.6532 & 0.6612 & 0.4522 & 0.4676 & 0.4857 & 0.4911 \\
\hline
SGH & 0.1758 & 0.1754 & 0.1786 & 0.1838 & 0.6309 & 0.6465 & 0.6503 & 0.6586 & 0.4759 & 0.4842 & 0.4818 & 0.4862 \\
\hline
DPLM & 0.1609 & 0.1776 & 0.1807 & 0.1835 & 0.6099 & 0.6261 & 0.6350 & 0.6422 & 0.4518 & 0.4744 & 0.4810 & 0.4865 \\
\hline
LSMH & 0.1738 & 0.1780 & 0.1788 & 0.1872 & 0.6369 & 0.6300 & 0.6483 & 0.6602 & 0.4643 & 0.4722 & 0.4892 & 0.4877 \\
\hline
\hline
CVH & 0.1676 & 0.1630 & 0.1601 & 0.1757 & 0.6026 & 0.6010 & 0.6094 & 0.6229 & 0.4447 & 0.4300 & 0.4233 & 0.4148 \\
\hline
CHMIS & 0.1507 & 0.1671 & 0.1640 & 0.1787 & 0.5628 & 0.5643 & 0.5643 & 0.5768 & 0.4404 & 0.4394 & 0.4341 & 0.4259 \\
\hline
IMH & 0.1663 & 0.1709 & 0.1742 & 0.1775 & 0.6285 & 0.6338 & 0.6454 & 0.6586 & 0.4475 & 0.4619 & 0.4634 & 0.4890 \\
\hline
LCMH & 0.1752 & 0.1784 & 0.1837 & 0.1809 & 0.6250 & 0.6339 & 0.6346 & 0.6349 & 0.4641 & 0.4726 & 0.4764 & 0.4777 \\
\hline
CMFH & 0.1678 & 0.1688 & 0.1705 & 0.1738 & 0.5846 & 0.6000 & 0.5956 & 0.6106 & 0.4703 & 0.4893 & 0.4972 & 0.4888 \\
\hline
DSTH & \textbf{0.2055} & \textbf{0.2012} & \textbf{0.2041} & \textbf{0.2040} & \textbf{0.6458} & \textbf{0.6603} & \textbf{0.6642} & \textbf{0.6692} & \textbf{0.5074} & \textbf{0.5089} & \textbf{0.5208}  & \textbf{0.5251}\\
\hline
\end{tabular}
\vspace{-2mm}
\end{table*}
\textbf{Complexity Analysis}.
The anchor graph construction includes anchor generation and distance computation between images and anchors. The time complexity of this process is $O(NKd_x)$. Solving discrete hash codes is conducted in an iterative process, the computational complexity is $O(\#iter(d_x N+d_yN + d_xL + d_yL + LN))$, where $\#iter$ denotes the number of iterations. Given $N\gg d_x(d_y)>L$, this process scales linearly with $N$. The computation of hash functions solves a linear system, whose time complexity is $O(N)$. Calculating hash codes of database images costs $O(N)$. Therefore, the whole offline learning consumes $O(N)$, which indicates the desirable scalability of the proposed DSTH. In online retrieval, generating hash codes for a query can be completed in $O((d_x+1)L)$.
\begin{figure*}
\centering
\mbox{
\subfigure{\includegraphics[width=42mm]{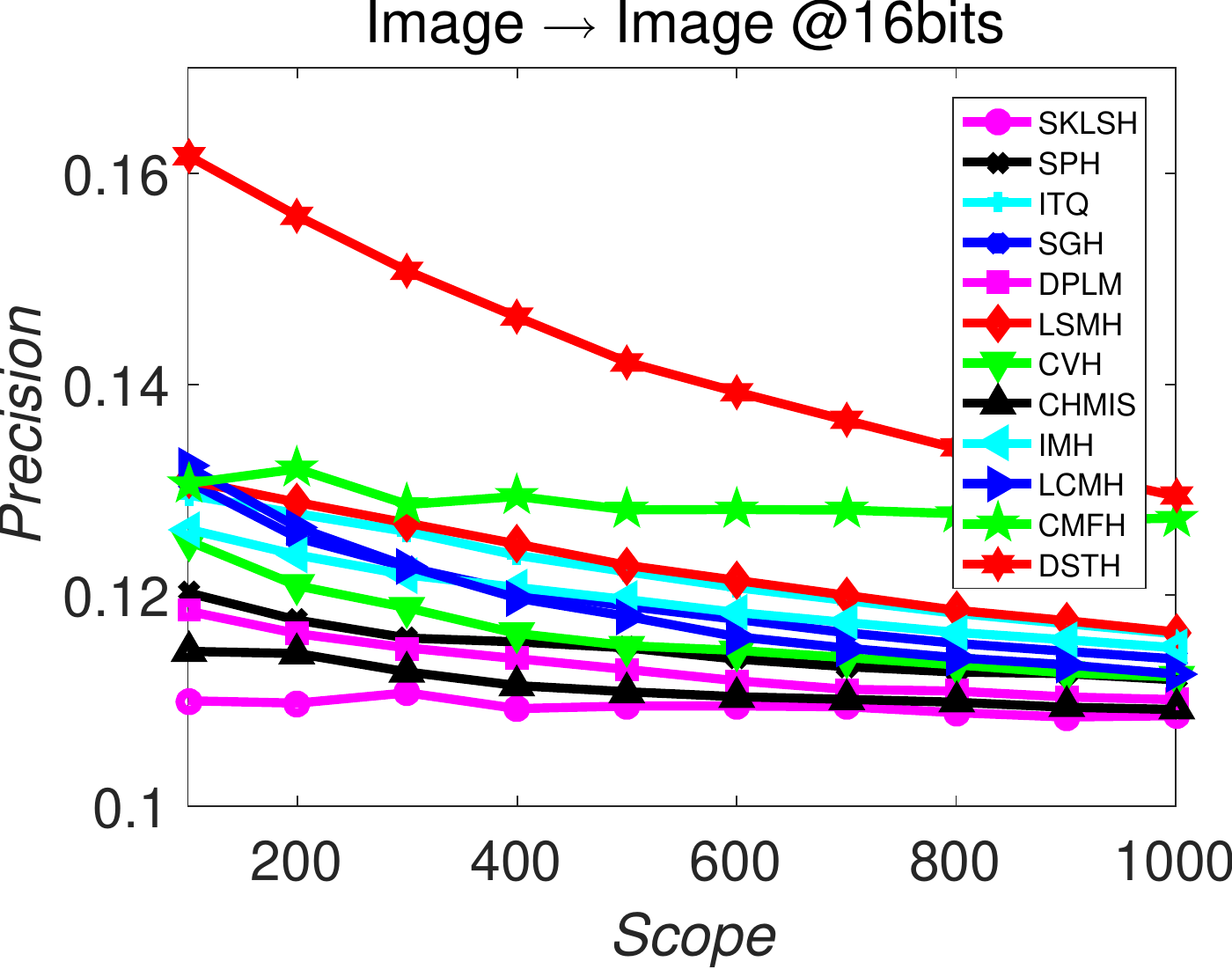}}
}\mbox{
\subfigure{\includegraphics[width=42mm]{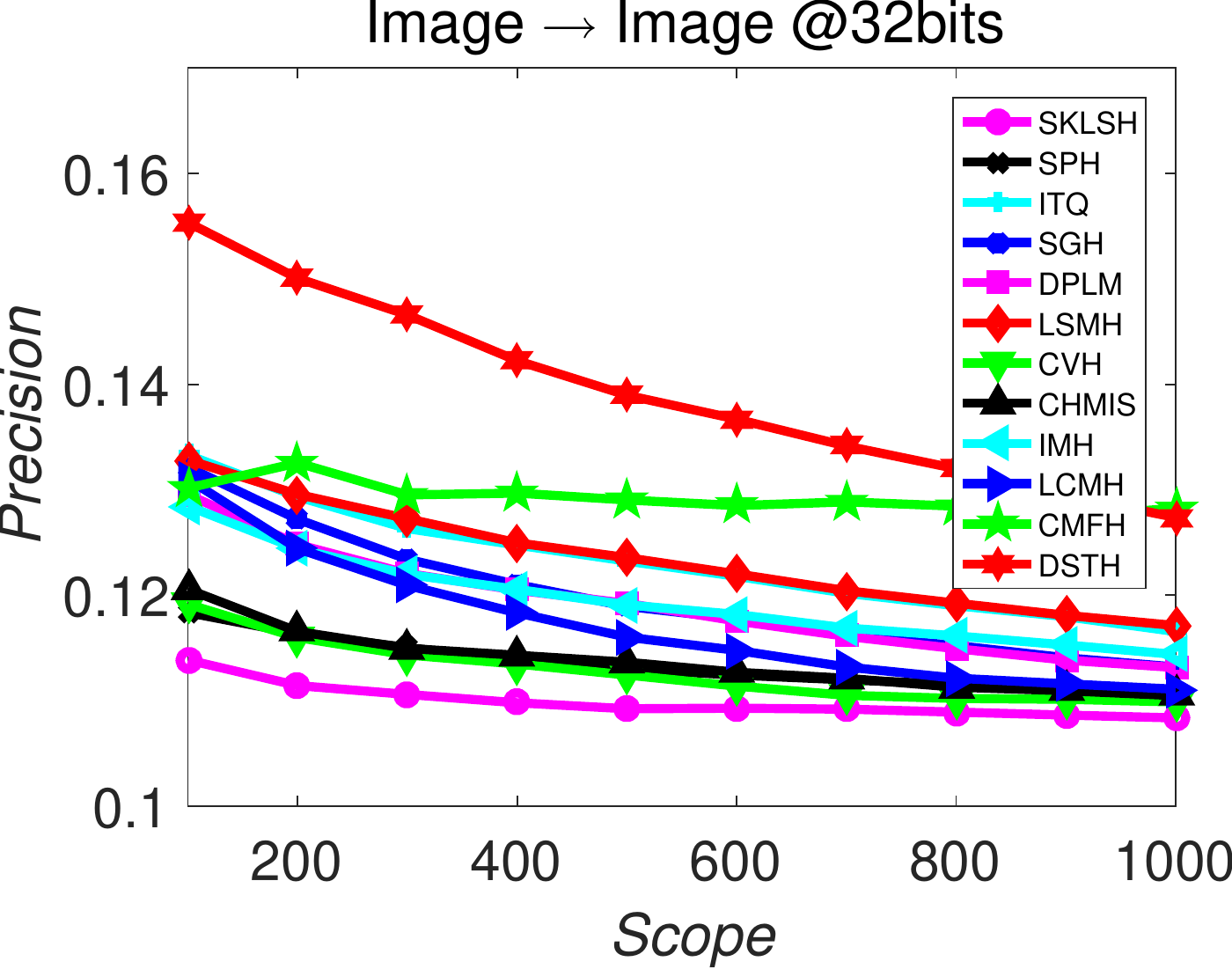}}
}\mbox{
\subfigure{\includegraphics[width=42mm]{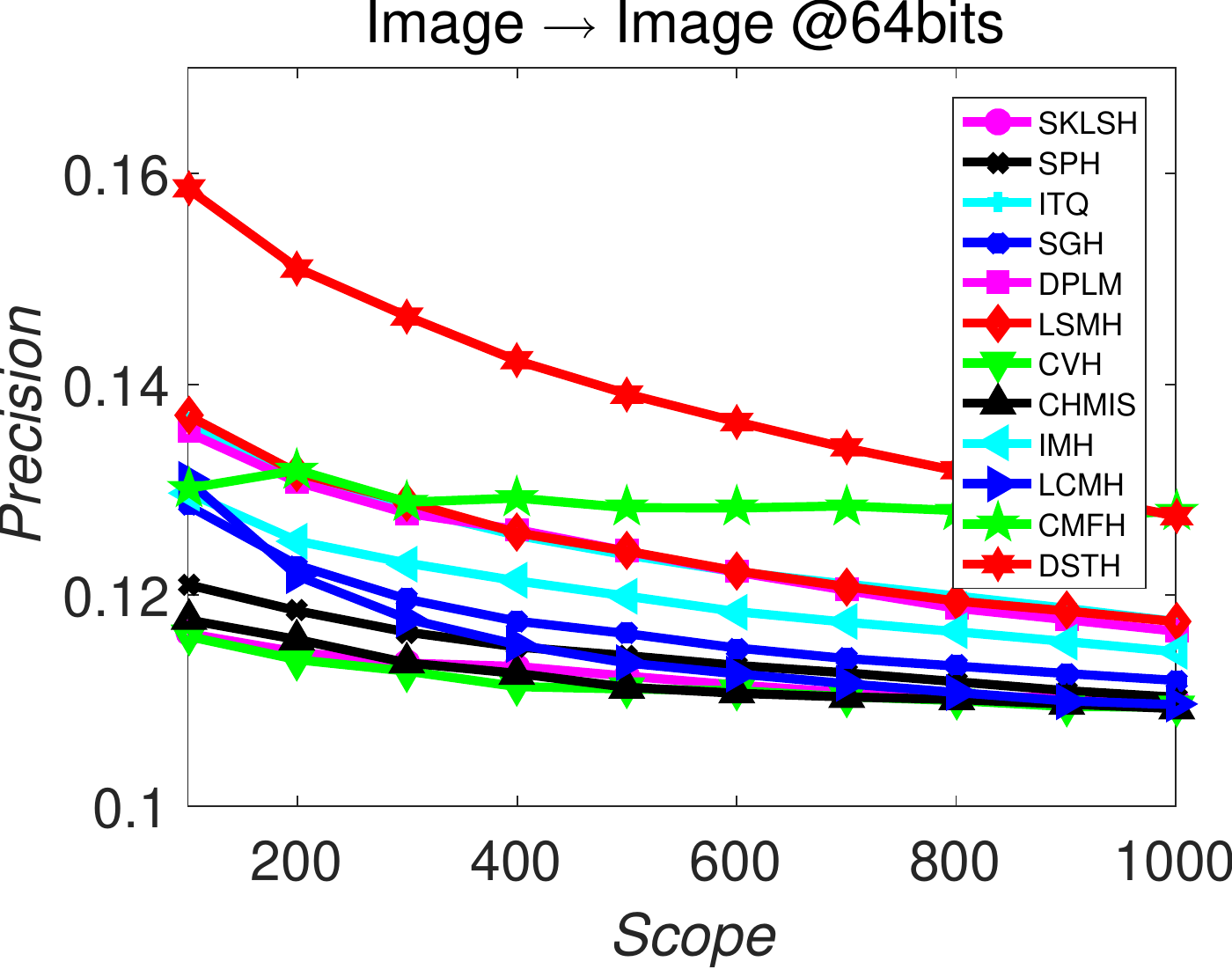}}
}\mbox{
\subfigure{\includegraphics[width=42mm]{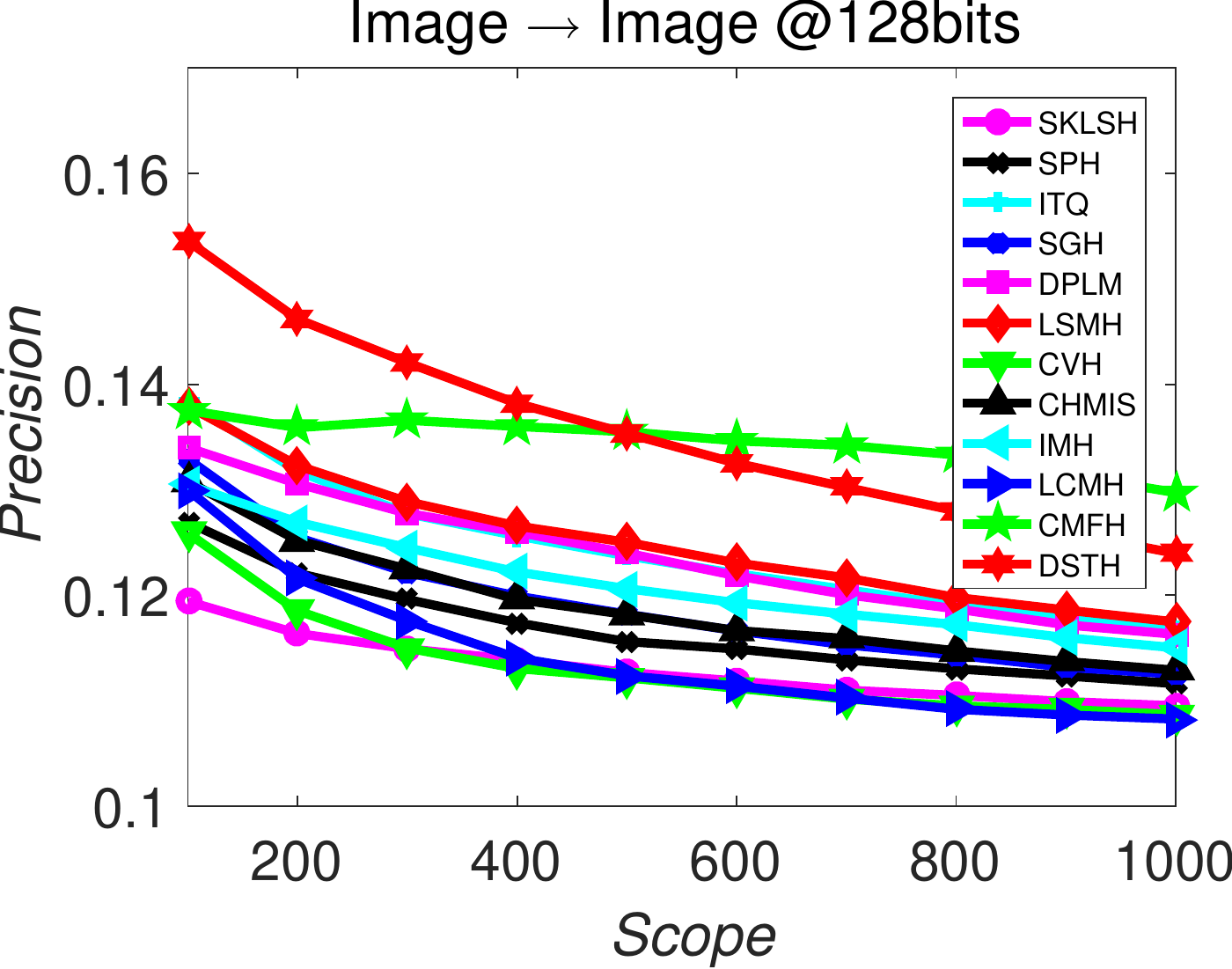}}
}
\caption{\emph{Precision-Scope} curves on \textbf{Wiki}.}
\label{fig_wiki}
\end{figure*}

\begin{figure*}
\centering
\mbox{
\subfigure{\includegraphics[width=42mm]{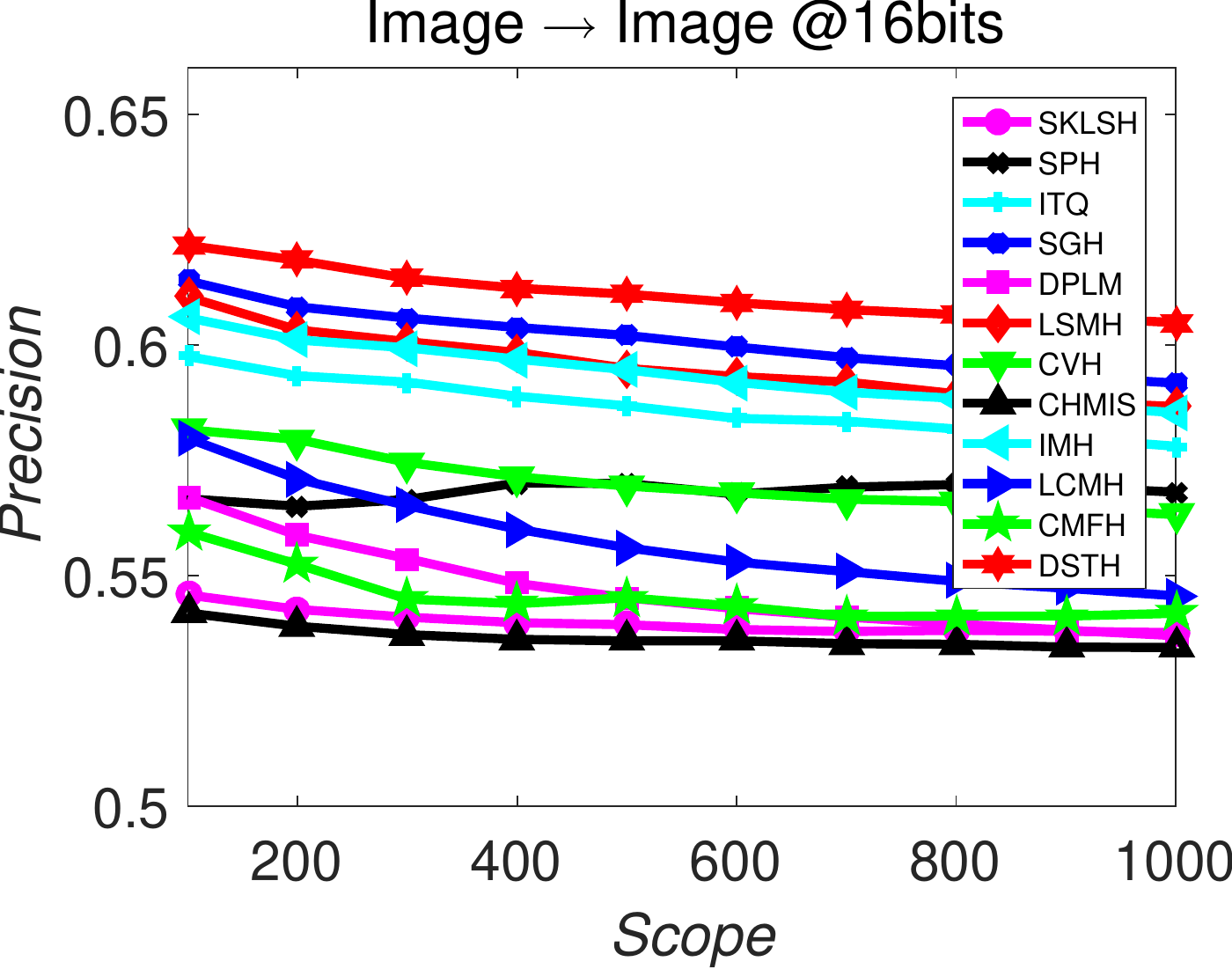}}
}\mbox{
\subfigure{\includegraphics[width=42mm]{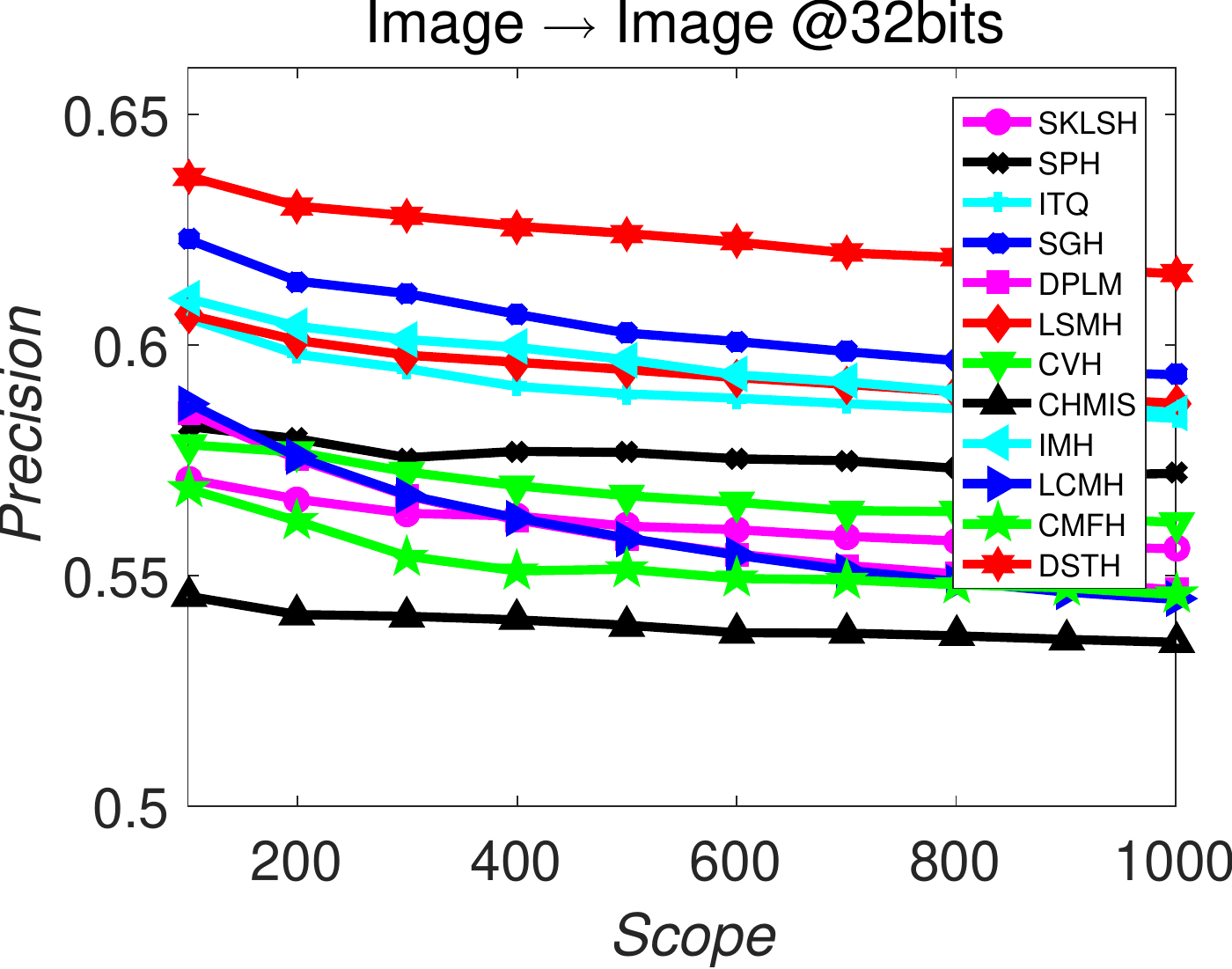}}
}\mbox{
\subfigure{\includegraphics[width=42mm]{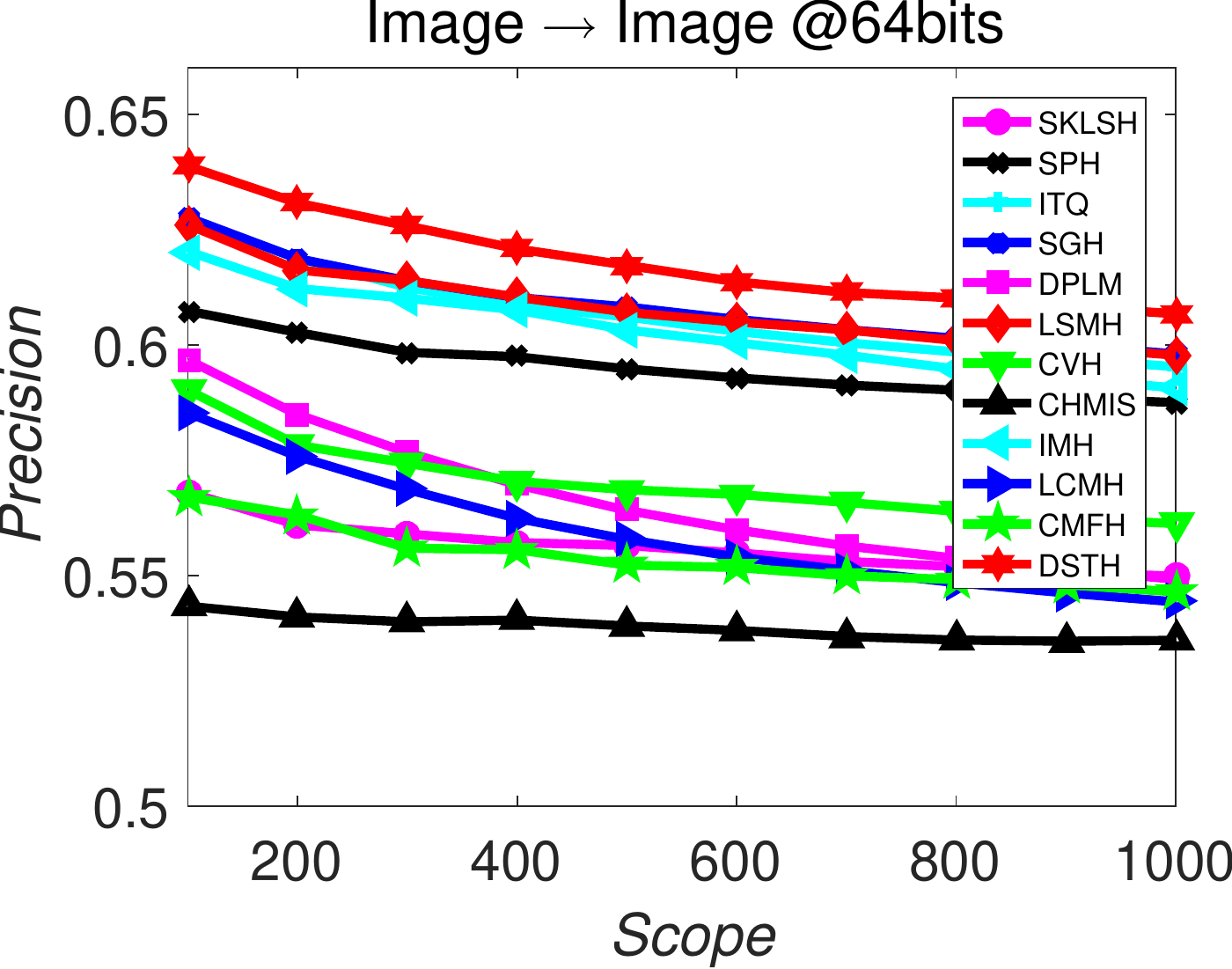}}
}\mbox{
\subfigure{\includegraphics[width=42mm]{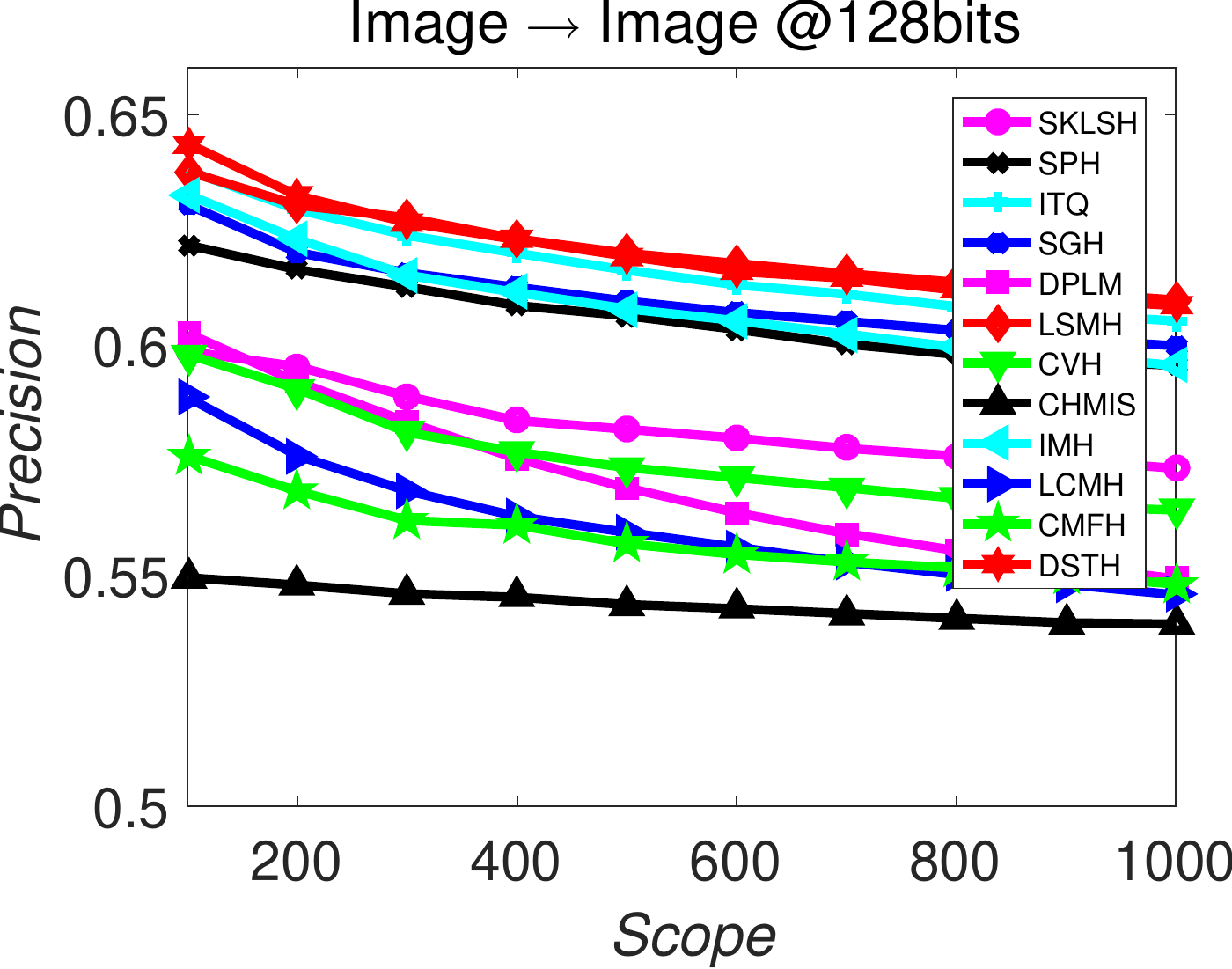}}
}
\caption{\emph{Precision-Scope} curves on \textbf{MIR Flickr}.}
\label{fig_mirflickr}
\vspace{-4mm}
\end{figure*}
\section{Experiments}
\label{sec:4}
In this section, we first introduce the experimental settings, including experimental dataset, evaluation metric and baselines.
Then, we present the comparison results with state-of-the-art approaches. Next, we evaluate the effects of
discrete optimization and semantic transfer. Finally, we give experimental results on convergence and parameter sensitiveness.
\begin{table}
\caption{Statistics of test collections (BoVW \cite{Sivic03} and BoW \cite{DBLP:journals/iet-ipr/ZhuJZF14} denote Bag-of-Visual-Words and Bag-of-Words respectively. LDA is Latent Dirichlet Allocation \cite{JMLRLDA})}
{
\centering
\begin{tabular}{|m{20mm}<{\centering}|m{14mm}<{\centering}|m{16mm}<{\centering}|m{17mm}<{\centering}|}
\hline
Datasets & \emph{\textbf{Wiki}} & \emph{\textbf{MIR Flickr}} & \emph{\textbf{NUS-WIDE}}\\
\hline
\#Database & 2,866 & 25,000 & 186,643\\
\hline
\#Query & 574 & 250 & 1,867 \\
\hline
\#Training & 2292 & 750 & 5,540\\
\hline
Visual Feature & BoVW\newline(128-D) & BoVW\newline(1000-D) & BoVW\newline(500-D)\\
\hline
Text Feature & BoW+LDA\newline(10-D) & BoW\newline(457-D) & BoW\newline(1000-D)\\
\hline
\end{tabular}}
\label{sdd}
\end{table}
\subsection{Experimental Dataset}
Experiments are conducted on three publicly available image datasets:
\textbf{Wiki}\footnote{http://www.svcl.ucsd.edu/projects/crossmodal/} \cite{MMWIKI}, \textbf{MIR Flickr}\footnote{http://lear.inrialpes.fr/people/guillaumin/data.php} \cite{MIRFLICKR} , and \textbf{NUS-WIDE}\footnote{http://lms.comp.nus.edu.sg/research/NUS-WIDE.htm} \cite{CIVRNUS}. All these datasets
are comprised of images and their contextual texts.
\begin{itemize}
  \item \textbf{Wiki} is composed of 2,866 multimedia documents which belong to 10 semantic categories. All the data are collected from Wikipedia\footnote{https://www.wikipedia.org/}. Each document contains an image and at least 70 textual words. In experiments, we represent visual contents with 128 dimensional SIFT histogram \cite{IJCVSIFT} and contextual text contents by 10 dimensional topic vector generated by latent Dirichlet allocation (LDA) \cite{JMLRLDA}.
 \item \textbf{MIR Flickr} consists of 25,000 images describing 38 semantic categories. This dataset is downloaded from the Flickr with its public API\footnote{https://www.flickr.com/services/api/}. Each image is associated with tags. The tags that appear less than 50 times are removed, and we finally obtain a vocabulary of 457 tags. In this work, visual contents of images in \textbf{\textbf{MIR Flickr}} are represented by 1000 dimensional dense SIFT histogram. Contents of contextual texts are represented by 457 dimensional binary vector, where each dimension indicates the presence of a tag.
  \item \textbf{NUS-WIDE} is composed of 269,648 images labelled by 81 concepts. Each image is also associated with tags. In experiments, we preserve 10 most common concepts and the corresponding 186,643 pairs. On \textbf{NUS-WIDE} dataset, we extract  500 dimensional SIFT histogram to describe the visual contents of images, and 1000 dimensional binary vector to represent the contextual texts.
\end{itemize}

Table \ref{sdd} summarizes the key statistics of
the test collections. For \textbf{Wiki}, as images are labelled by 10 independent categories,
images in this dataset are considered to be relevant only if they belong to the same category.
For \textbf{MIR Flickr} and \textbf{NUS-WIDE}, images are labelled by several tags,
and therefore images are considered to be relevant if they share at least one tag.

\begin{table*}
\caption{Effects of discrete optimization and semantic transfer. DSTH-I denotes the variant approach that removes discrete constraint. It adopts conventional two-step relaxing+rounding optimization to solve hash codes. DSTH-II denotes the variant approach that removes bit-balance constraint. DSTH-III denotes the variant approach that removes bit−uncorrelation constraint. DSTH-IV denotes the approach which only considers visual similarity preservation.}
\label{discretop}
\centering
\begin{tabular}{|p{13.5mm}<{\centering}|p{8.5mm}<{\centering}|p{8.5mm}<{\centering}|p{8.5mm}<{\centering}|p{10mm}<{\centering}
|p{8.5mm}<{\centering}|p{8.5mm}<{\centering}|p{8.5mm}<{\centering}|p{10mm}<{\centering}|p{8.5mm}<{\centering}|p{8.5mm}<{\centering}
|p{8.5mm}<{\centering}|p{10mm}<{\centering}|}
\hline
\multirow{2}{*}{Methods} & \multicolumn{4}{c|}{\emph{\textbf{Wiki}}} & \multicolumn{4}{c|}{\emph{\textbf{MIR Flickr}}} & \multicolumn{4}{c|}{\emph{\textbf{NUS-WIDE}}}\\
\cline{2-13}
& 16 & 32 & 64  & 128  & 16 & 32 & 64  & 128 & 16 & 32 & 64 & 128 \\
\hline
DSTH-I & 0.1675 & 0.1740 & 0.1725 & 0.1719 & 0.6326 & 0.6343 & 0.6554 & 0.6604 & 0.4963 & 0.5075 & 0.5150 & 0.5181 \\
\hline
DSTH-II & 0.2018 & 0.1997 & 0.2037 & 0.2018 & 0.6425 & 0.6450 & 0.6610 & 0.6634 & 0.4844 & 0.5015 & 0.5133 & 0.5231 \\
\hline
DSTH-III & 0.1981 & 0.1984 & 0.1924 & 0.1972 & 0.6349 & 0.6234 & 0.6540 & 0.6616 & 0.4685 & 0.4784 & 0.4922 & 0.4991 \\
\hline
DSTH-IV & 0.1998 & 0.1939 & 0.2001 & 0.1979 & 0.6301 & 0.6362 & 0.6323 & 0.6246 & 0.4765 & 0.4744 & 0.4650 & 0.4512 \\
\hline
DSTH & \textbf{0.2055} & \textbf{0.2012} & \textbf{0.2041} & \textbf{0.2040} & \textbf{0.6458} & \textbf{0.6603} & \textbf{0.6642} & \textbf{0.6692} & \textbf{0.5074} & \textbf{0.5089} & \textbf{0.5208}  & \textbf{0.5251}\\
\hline
\end{tabular}
\end{table*}
\begin{table}
\caption{Performance variations of DSTH with the training size on \textbf{NUS-WIDE} when hash code length is 128.}
\label{trainingsize}
\centering
\begin{tabular}{|p{19mm}<{\centering}|p{8mm}<{\centering}|p{8mm}<{\centering}|p{8mm}<{\centering}|p{8mm}<{\centering}|p{8mm}<{\centering}|}
\hline
\#Training & 0.5K & 1.5K & 2K & 2.5K & 3K \\
\hline
mAP & 0.5006 & 0.5010 & 0.5106 & 0.5126 & 0.5202\\
\hline
\#Training & 3.5K & 4K & 4.5K & 5K & 5540 \\
\hline
mAP & 0.5131 & 0.5183 & 0.5241 & 0.5256 & 0.5251\\
\hline
\end{tabular}
\end{table}

\subsection{Evaluation Metric}
In experiments, mean average precision (mAP) \cite{TIPCMFH,SGH} is adopted as the evaluation metric.
For a given query, average precision (AP) is calculated as
\begin{equation}
\begin{aligned}
AP=\frac{1}{NR}\sum_{r=1}^R Pre(r)\delta(r)
\end{aligned}
\end{equation}
where $R$ is the total number of retrieved images, $NR$ is the number of relevant images in the retrieved set, $Pre(r)$ denotes the precision of top $r$ retrieval images, which is defined as the ratio between the number of the relevant images and the number of retrieved images $r$, and $\delta(r)$ is indicator function which equals to 1 if the $r_{th}$ image is relevant to query, and vice versa. In experiments, we set the total number of retrieved images as 100 to report experimental results. Furthermore,
\emph{Precision-Scope} curve is also plotted to illustrate the retrieval performance variations with respect to the number of retrieved images.
\begin{figure}
\centering
\includegraphics[width=85mm]{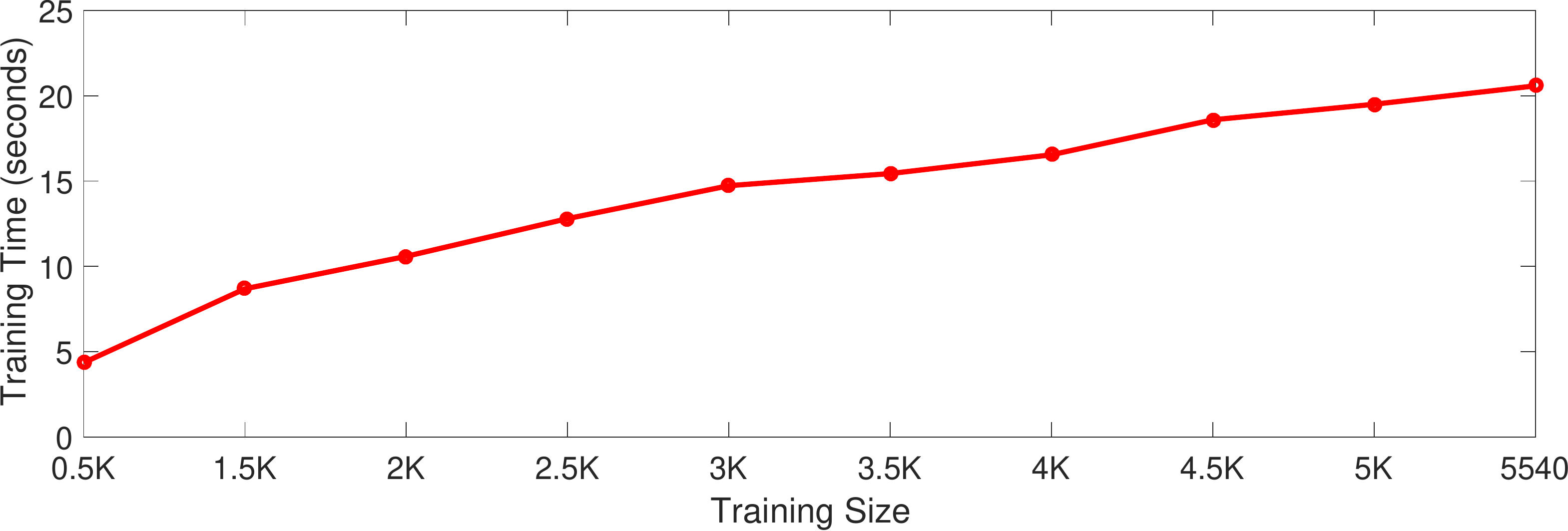}
\caption{Variations of training time with training data size.}
\label{fig:trainsize}
\end{figure}
\subsection{Evaluation Baselines}
We compare DSTH with several state-of-the-art uni-modal hashing approaches, which can be used to support SCBIR\footnote{Multi-modal hashing methods \cite{CHMIS,MFH,MVLH} are not used for comparison because they need both images and texts as query. Their retrieval scenarios are different from SCBIR.}. They include:
\begin{enumerate}[1.]
  \item \textbf{Shift-invariant kernel locality sensitive hashing} (\textbf{SKLSH}) \cite{SKLSH}. It is a representative data-independent hashing, which generates hash codes by random projections with distribution-free encoding.
  \item \textbf{Spectral hashing} (\textbf{SPH}) \cite{SPH}. Hash codes are computed by eigenvalue decomposition on visual Laplacian matrix. Hash functions are constructed with an efficient Nystrom method.
  \item \textbf{Iterative quantization} (\textbf{ITQ}) \cite{ITQ}. ITQ first learns relaxed hash codes with principal component analysis (PCA) \cite{jolliffe2002principal}. Then, it generates the hash codes by minimizing the quantization errors with optimal iterative rotation.
  \item \textbf{Scalable graph hashing} (\textbf{SGH}) \cite{SGH}. SGH leverages feature transformation to approximate the visual graph, and thus avoids explicit similarity graph computing. In this method, hash functions are learned in a bit-wise manner with a sequential learning.
  \item \textbf{Latent semantic minimal hashing} (\textbf{LSMH}) \cite{LSMH}. Minimum encoding and matrix factorization are combined together to simultaneously learn latent semantic feature which refines original features, and hash codes.
  \item \textbf{Discrete proximal linearized minimization} (\textbf{DPLM}) \cite{TIP2016binary}. We use unsupervised setting of DPLM for comparison. This method directly handles with discrete constraint. The hash codes are solved by iterative procedures with each iteration admitting an analytical solution.
\end{enumerate}
\begin{figure*}
\centering
\mbox{
\subfigure{\includegraphics[width=42mm]{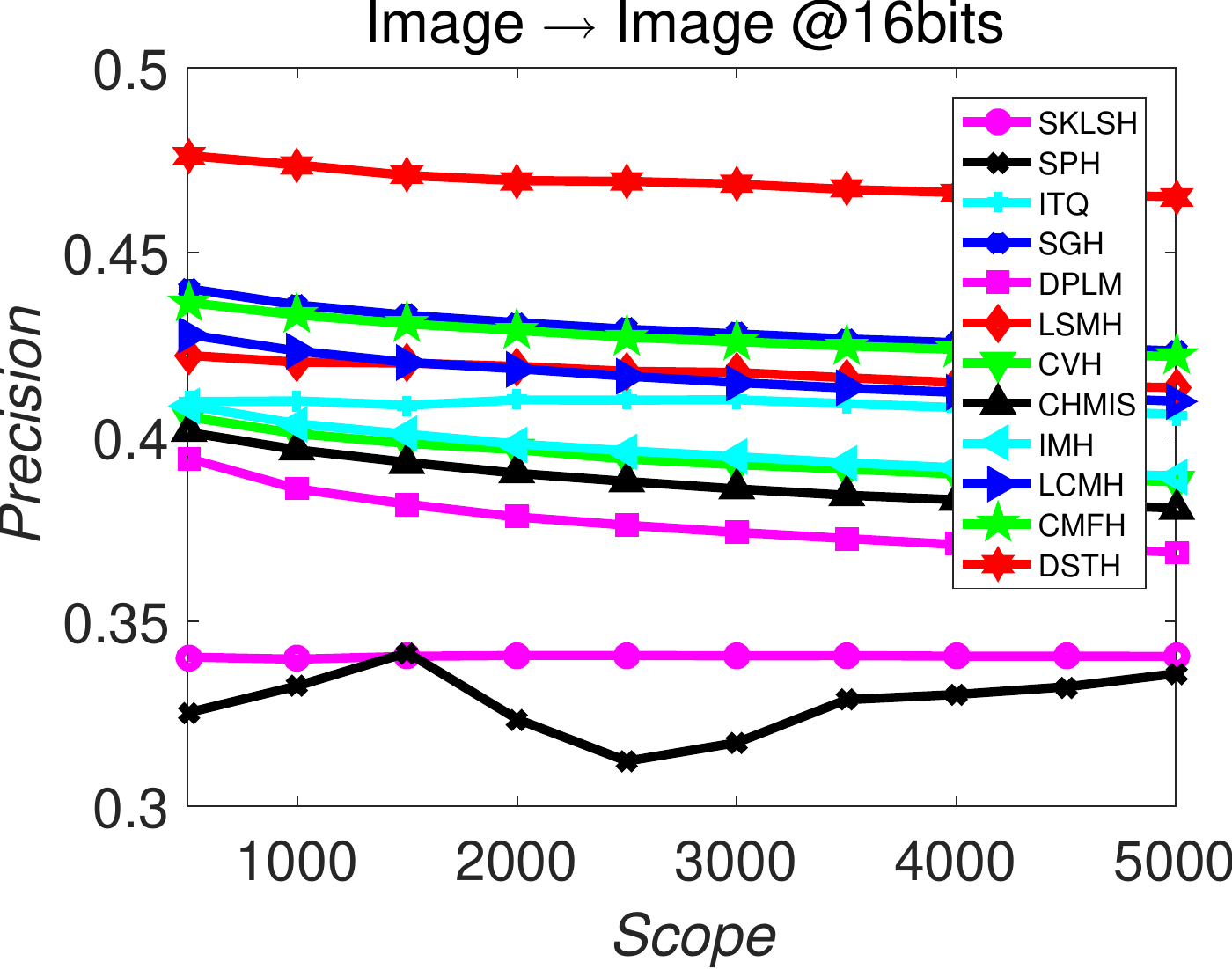}}
}\mbox{
\subfigure{\includegraphics[width=42mm]{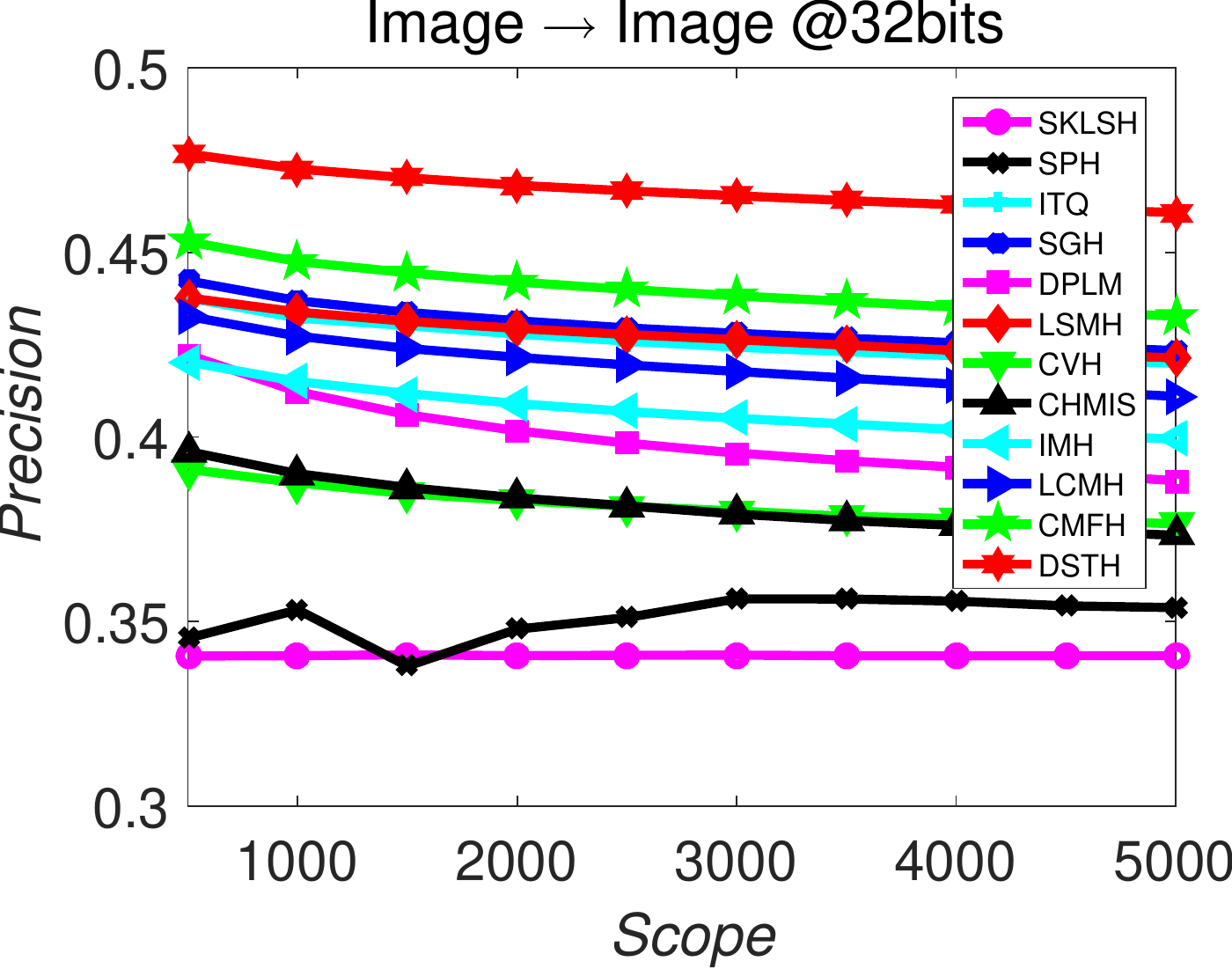}}
}\mbox{
\subfigure{\includegraphics[width=42mm]{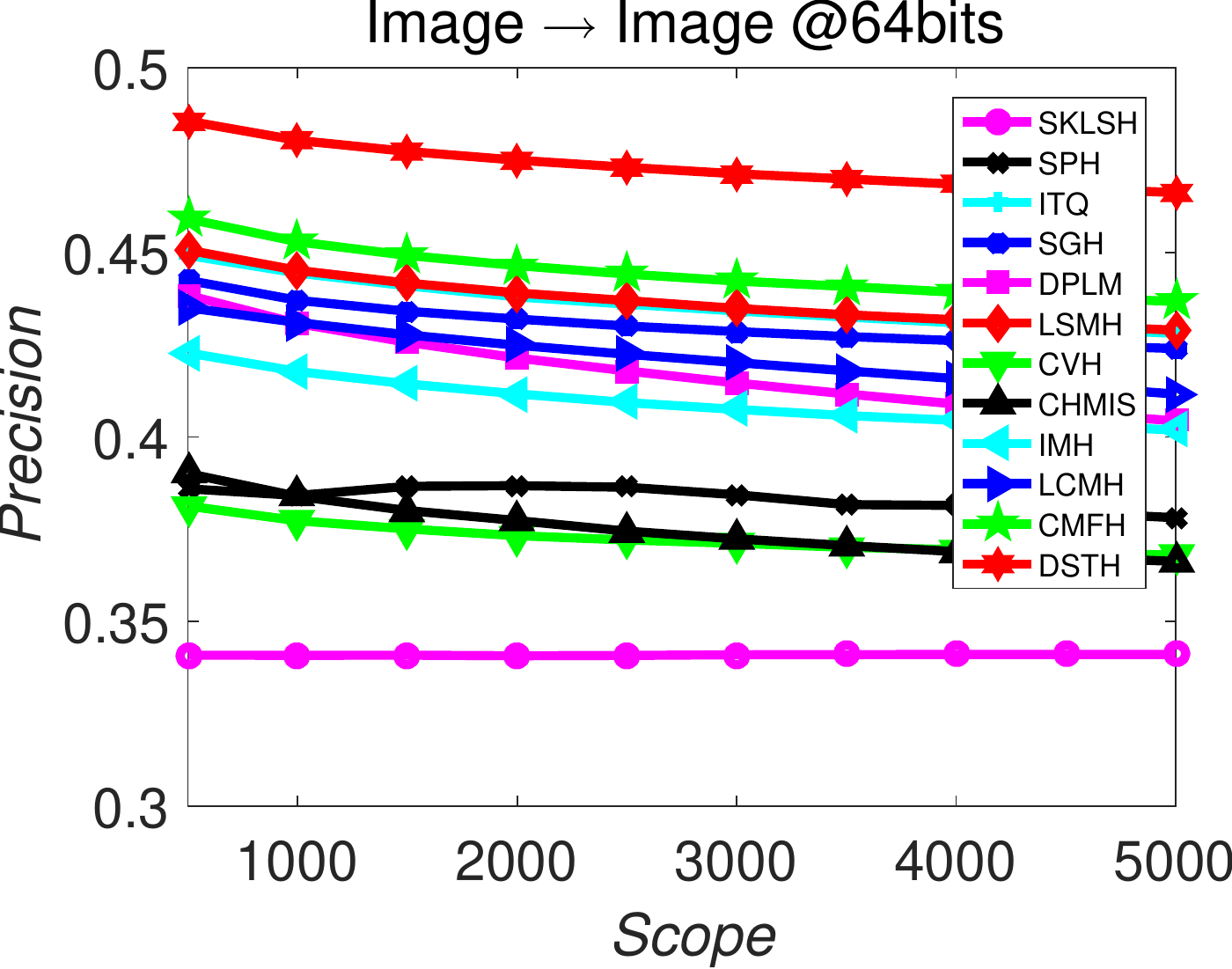}}
}\mbox{
\subfigure{\includegraphics[width=42mm]{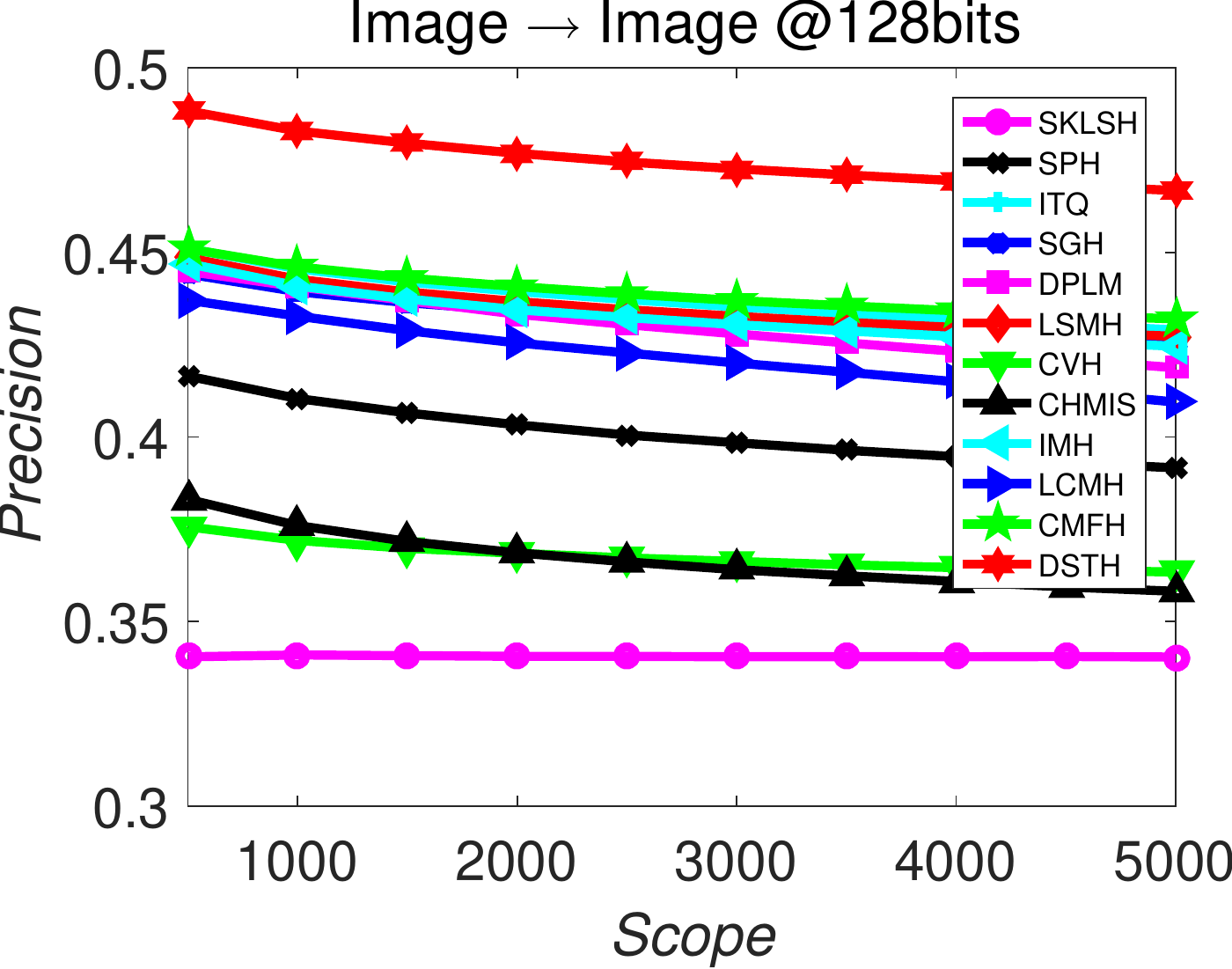}}
}
\caption{\emph{Precision-Scope} curves on \textbf{NUS-WIDE}.}
\label{fig_nuswide}
\end{figure*}

Since cross-modal hashing can also be used for SCBIR, we also incorporate several state-of-the-art cross-modal hashing methods for comparison. They include\footnote{For cross-modal hashing approaches, as we evaluate the performance of SCBIR, we only use hash codes of images.}.
:
\begin{enumerate}[1.]
  \item \textbf{Cross-view hashing} (\textbf{CVH}) \cite{CMFIJCAI}. CVH extends spectral hashing to learn hash functions by jointly
	minimizing Hamming distances of similar samples and maximizing that of dissimilar samples.
  \item \textbf{Composite hashing with multiple information sources} (\textbf{CHMIS}) \cite{CHMIS}. It integrates
	discriminative information from several heterogeneous modalities into the hash codes with
	proper weights. For comparison fairness, text input is removed and only visual input is preserved
	in CHMIS.
  \item \textbf{Inter-media hashing} (\textbf{IMH}) \cite{CMFSIGMOD}. IMH formulates hash function learning in a
	framework where intra-similarity of each individual modal and inter-correlations between different
	modalities are both preserved in hash codes.
  \item \textbf{Linear cross-modal hashing} (\textbf{LCMH}) \cite{LCMH}. In this method, intra-modality similarity is approximately preserved with the new representations
    of samples which are calculated as the distances to several centroids of the clusters. The inter-modality correlation is preserved via
    the shared binary subspace learning.
  \item \textbf{Collective matrix factorization hashing} (\textbf{CMFH}) \cite{TIPCMFH}. CMFH performs cross-modal similarity
	search in a latent shared semantic space by collective matrix factorization.
\end{enumerate}

All parameters in the compared approaches are adjusted according to the relevant literatures and report the best performance.  For implementation of CVH, we kindly use the source code provided by \cite{MLBE}. For LCMH, we carefully implement the code according to relevant paper. For SPH, SKLSH, ITQ, SGH, DPLM, CHMIS, IMH, and CMFH, we directly download the implementation codes from authors' websites.
\subsection{Implementation Details}
5-fold cross validation is adopted to choose parameters. The best performance of DSTH is achieved when $k$ is set to 5, 5, 8 on three datasets respectively (Three datasets denote \textbf{Wiki}, \textbf{MIR Flickr}, and \textbf{NUS-WIDE} successively. Please find the same below.). Furthermore, DSTH has parameters: $\alpha$ and $\beta$. They control the processes of semantic discovery and transfer. The best performance is achieved when $\{\alpha=0.0001, \beta=10000\}$, $\{\alpha=100, \beta=10000\}$, and $\{\alpha=0.0001, \beta=100\}$ on three datasets respectively. The parameters $\mu$ and $\rho$ are used for ALM optimization. The optimal performance is obtained when $\{\mu=1, \rho=2\}$, $\{\mu=0.01, \rho=2\}$, and $\{\mu=0.0001, \rho=2\}$ on three datasets respectively. $\eta$ is used to learn hash functions. The best $\eta$ is set to 0.1, 1000, and 100 on three datasets, respectively.

\begin{table}
\caption{Training time (seconds) of all evaluated approaches on \textbf{NUS-WIDE} when hash code length is 128.}
\label{runtimetable}
\centering
\begin{tabular}{|p{16mm}<{\centering}|p{8mm}<{\centering}|p{8mm}<{\centering}|p{8mm}<{\centering}|p{11mm}<{\centering}|}
\hline
\multirow{2}{*}{Methods} & \multicolumn{4}{c|}{\emph{\textbf{NUS-WIDE}}}\\
\cline{2-5}
& 16 & 32 & 64  & 128  \\
\hline
SKLSH & 0.01 & 0.01 & 0.01 & 0.01 \\
\hline
SPH & 0.11 & 0.13 & 0.21 & 0.54 \\
\hline
ITQ & 0.22 & 0.27 & 0.38 & 0.82 \\
\hline
SGH & 2.84  & 2.85 & 2.85 & 2.94 \\
\hline
DPLM & 1.43 & 1.31 & 1.42 & 1.40 \\
\hline
LSMH & 0.28 & 0.37 & 0.67 & 1.27 \\
\hline
\hline
CVH & 0.37 & 0.38 & 0.42 & 0.72 \\
\hline
CHMIS & 119.77 & 157.74 & 219.48 & 302.90 \\
\hline
IMH & 43.95 & 43.85 & 43.97 & 44.23 \\
\hline
LCMH & 5.23 & 5.14 & 5.00  & 5.20 \\
\hline
CMFH & 12.96 & 14.05 & 16.01 & 20.63 \\
\hline
DSTH & 12.46 & 16.55 & 19.17 & 20.58 \\
\hline
\end{tabular}
\vspace{-2mm}
\end{table}

\subsection{Comparison Results}
Table \ref{resulttable} presents the main mAP results. Code length on all datasets is varied in the range of $\{16, 32, 64, 128\}$. Figure \ref{fig_nuswide} reports \emph{Precision-Scope} curve on \textbf{NUS-WIDE}. The search scope is ranged from 500 to 5000 with stepsize 500. The presented results clearly demonstrate that DSTH consistently outperforms the compared approaches on all datasets and hashing bits. On \textbf{Wiki} and \textbf{NUS-WIDE}, DSTH outperforms the second best performance by more than 2\%. Among the competitors, SKLSH achieves the worst performance in most cases. This is because SKLSH is data-independent hashing which generates hash codes without integrating any semantics from retrieved images. In addition, it is interesting to find that cross-modal hashing may not obtain better performance than the uni-modal hashing approaches in many cases. This experimental phenomenon can be explained as follows: cross-modal hashing is specially developed for cross-modal retrieval task. Hence, discovering the shared semantics is the main design target. This design may be propitious to the target of cross-modal retrieval. However, the shared semantic space may loss valuable visual discriminative information, which is essential for the task of SCBIR. Therefore, the retrieval performance of cross-modal hashing methods on SCBIR may be impaired.

We also investigate the training time of all evaluated approaches. This experiment is conducted on \textbf{NUS-WIDE} when hash code length is fixed to 128. The running time is recoded on a PC with Intel(R) Xeon(R) CPU E5-1650@3.60GHz and 64GB RAM. Table \ref{runtimetable} presents the main results. We can find that the time consumption of DSTH is on the same
order of magnitude as that of CMFH. The time cost of DSTH is acceptable. It is much faster than IMH and CHMIS.
\begin{figure*}
\centering
\mbox{
\subfigure[$\alpha$ is fixed to 0.0001]{\includegraphics[width=38mm]{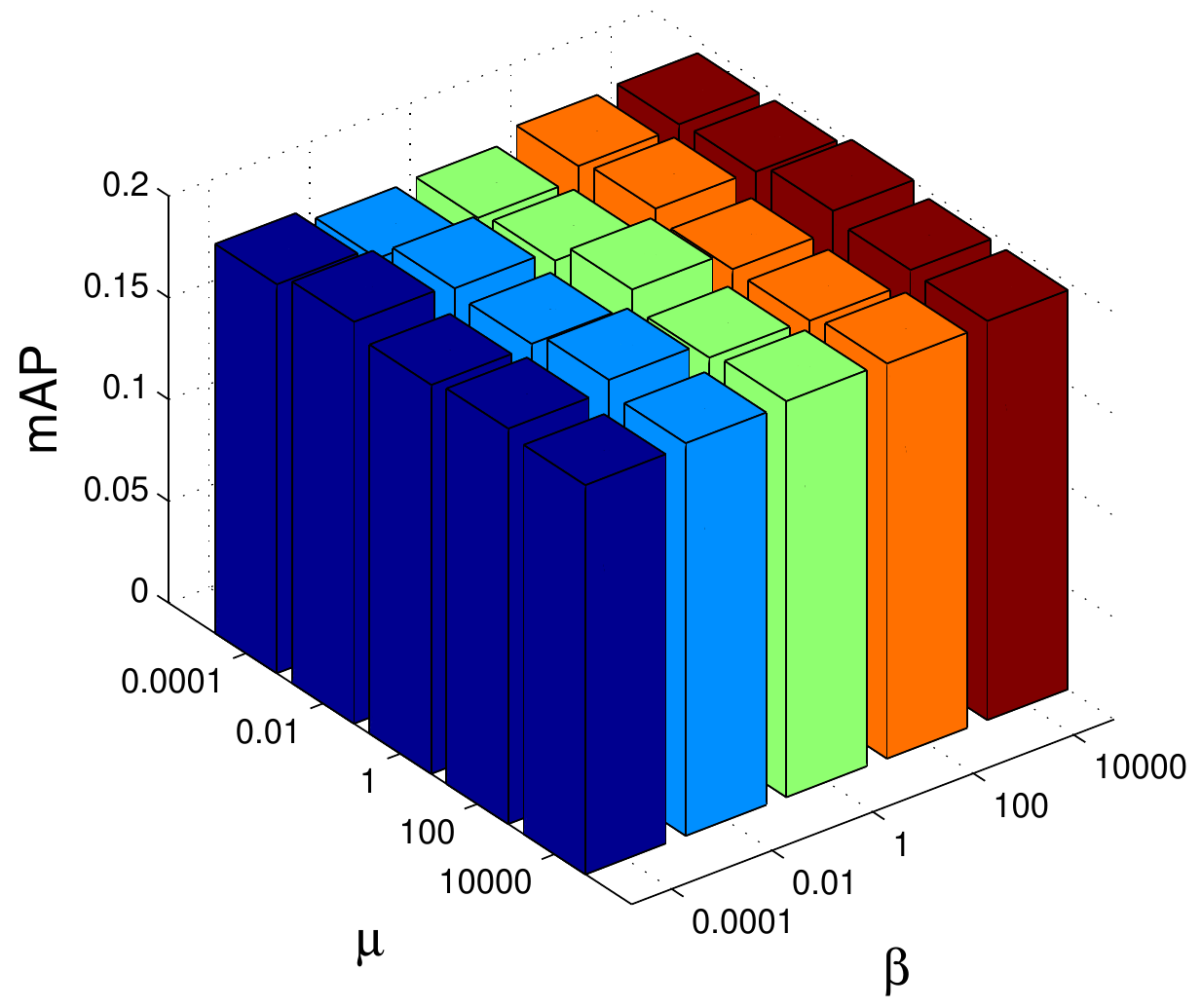}}
}\mbox{
\subfigure[$\beta$ is fixed to 10000]{\includegraphics[width=38mm]{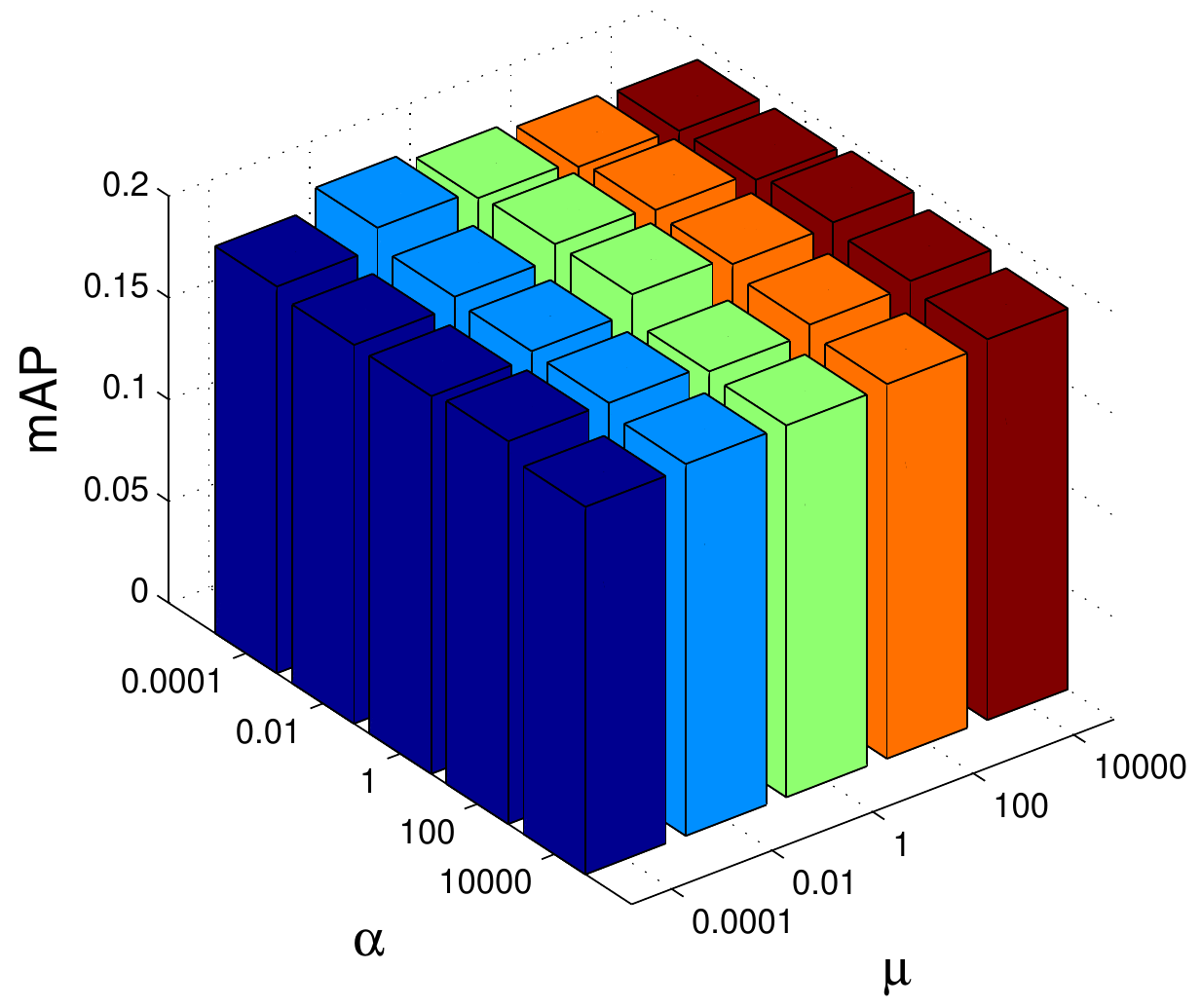}}
}\mbox{
\subfigure[$\mu$ is fixed to 1]{\includegraphics[width=38mm]{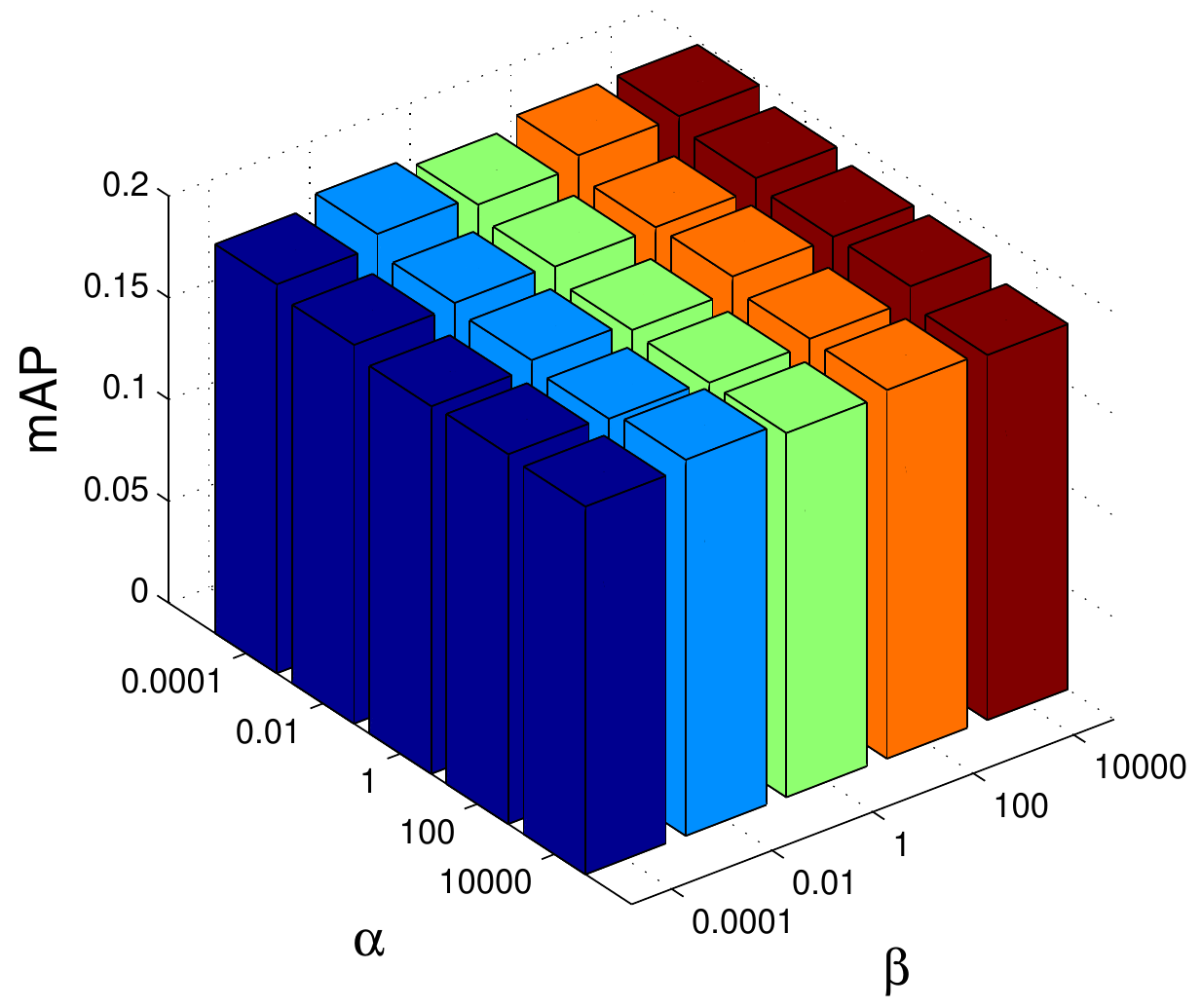}}
}\mbox{
\subfigure[]{\includegraphics[width=38mm]{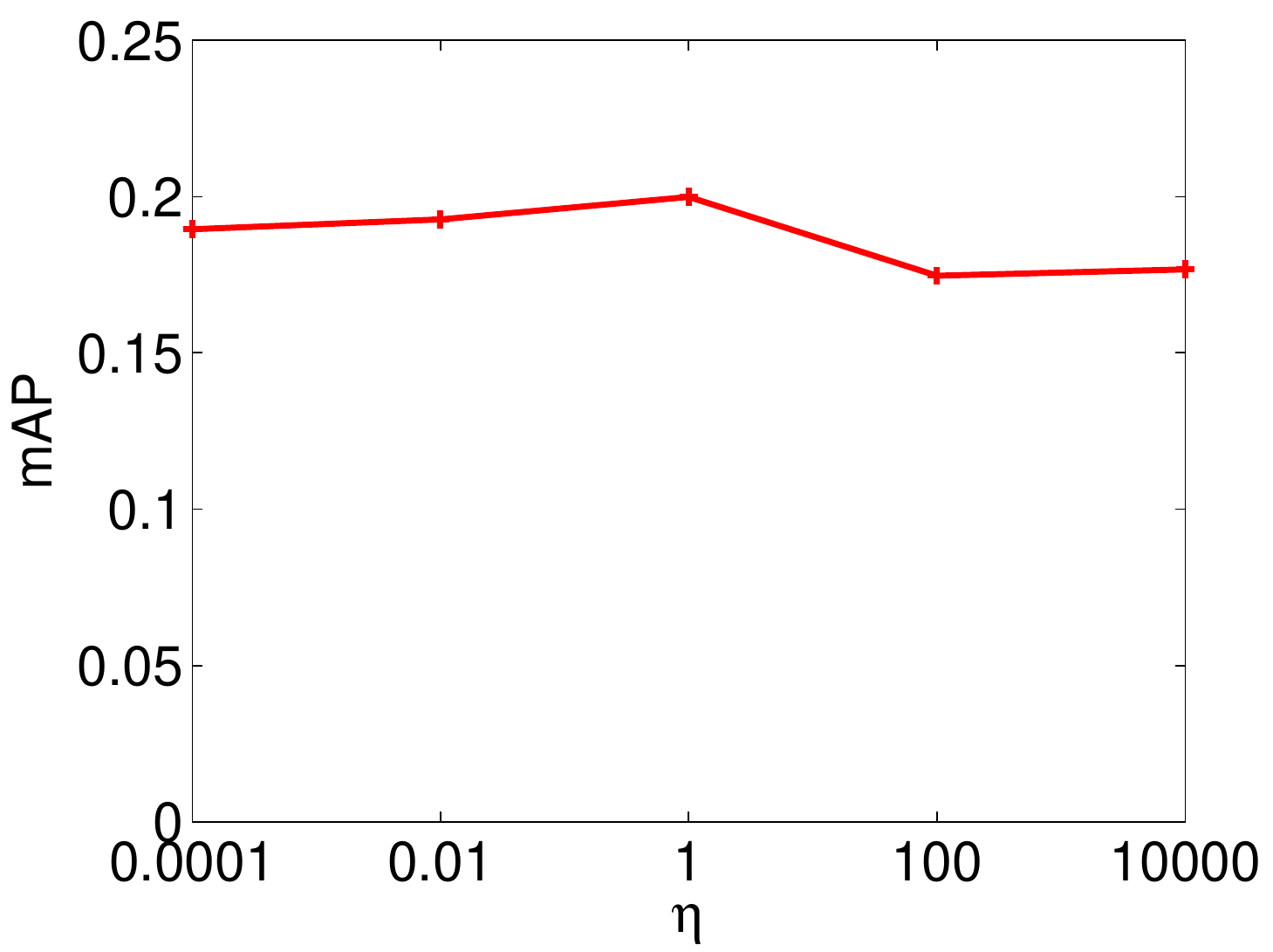}}
}
\caption{Performance variations with parameters on \textbf{Wiki}.}
\label{wiki_para}
\end{figure*}
\begin{figure*}
\centering
\mbox{
\subfigure[$\alpha$ is fixed to 100]{\includegraphics[width=38mm]{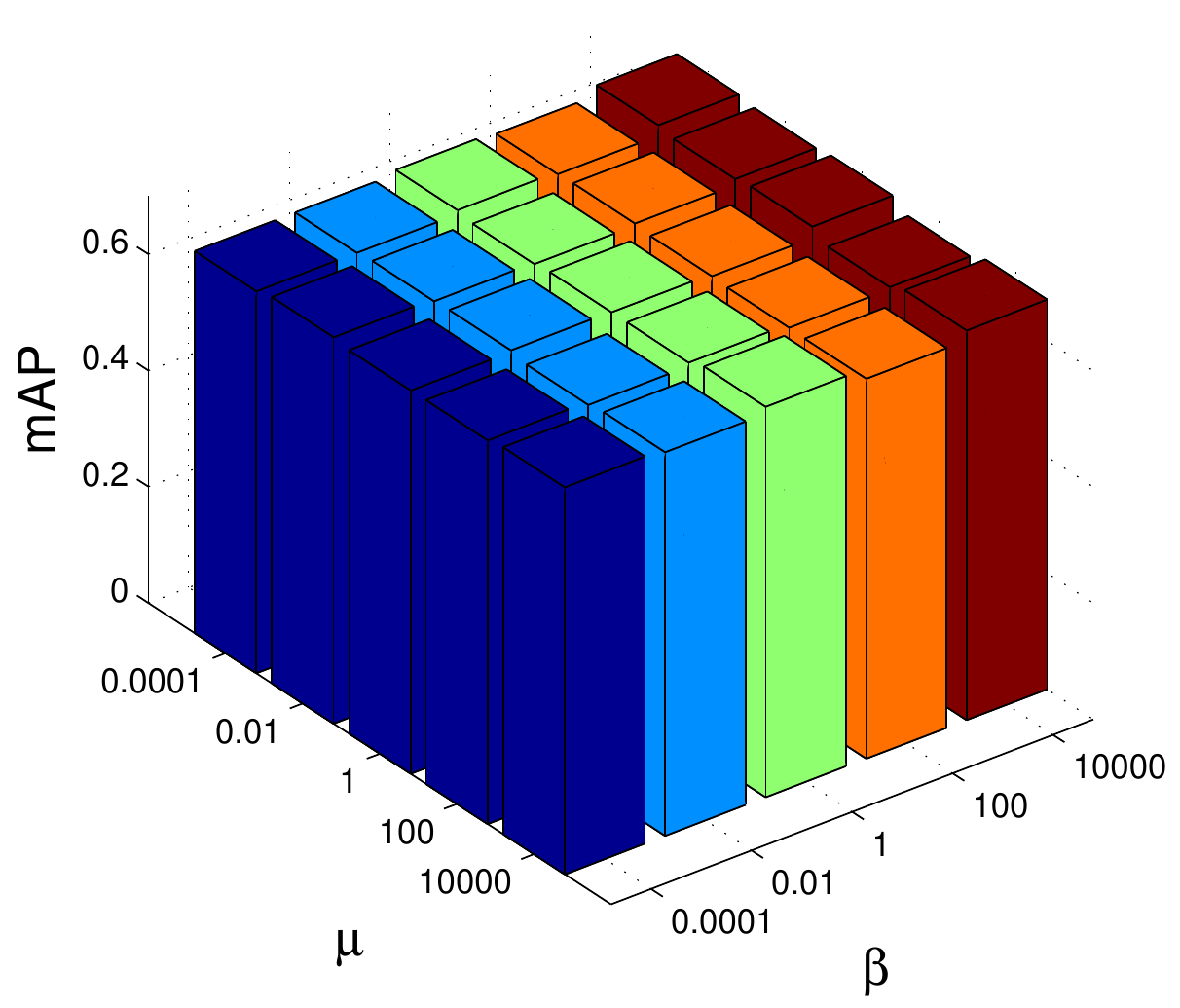}}
}\mbox{
\subfigure[$\beta$ is fixed to 10000]{\includegraphics[width=38mm]{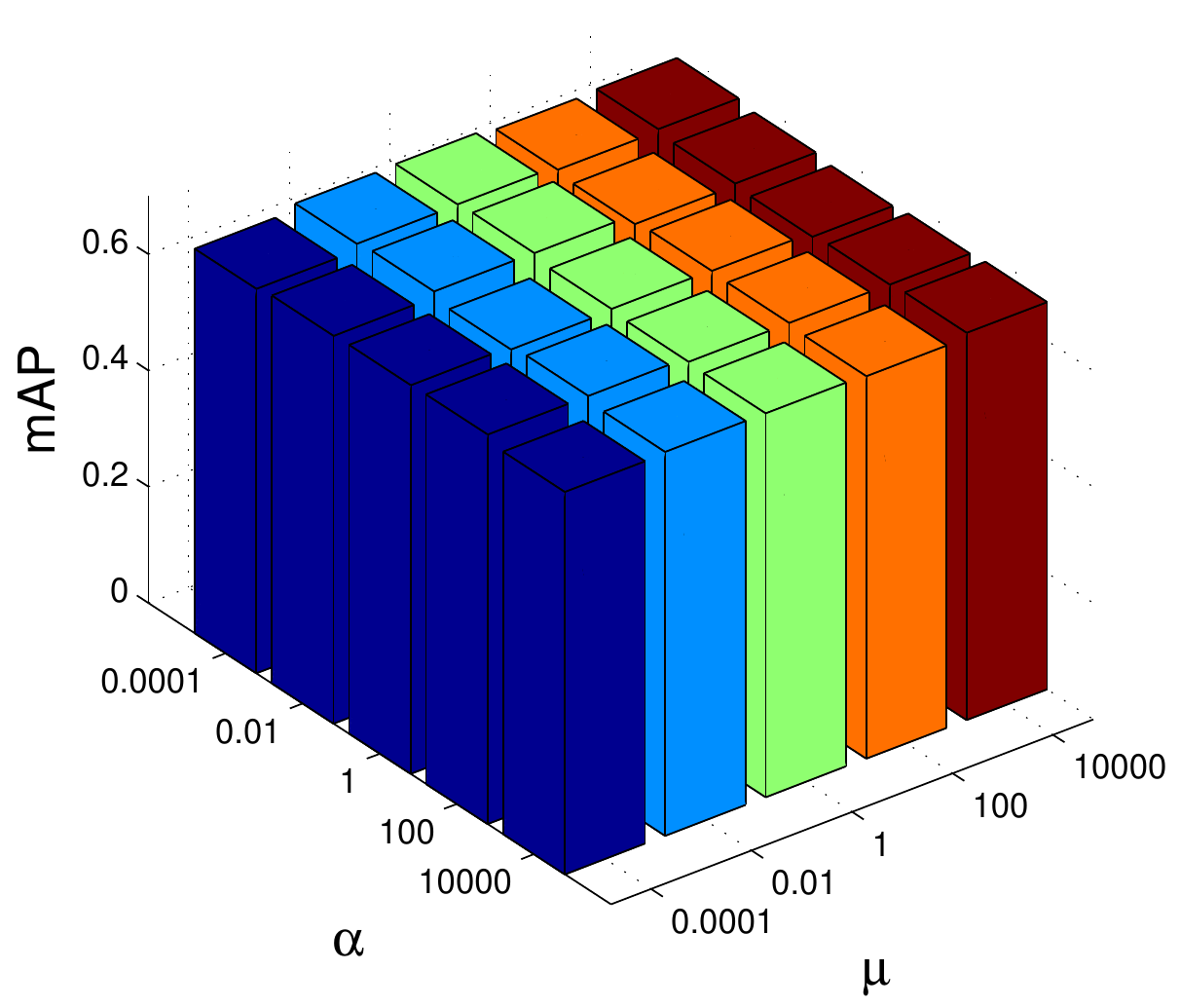}}
}\mbox{
\subfigure[$\mu$ is fixed to 0.01]{\includegraphics[width=38mm]{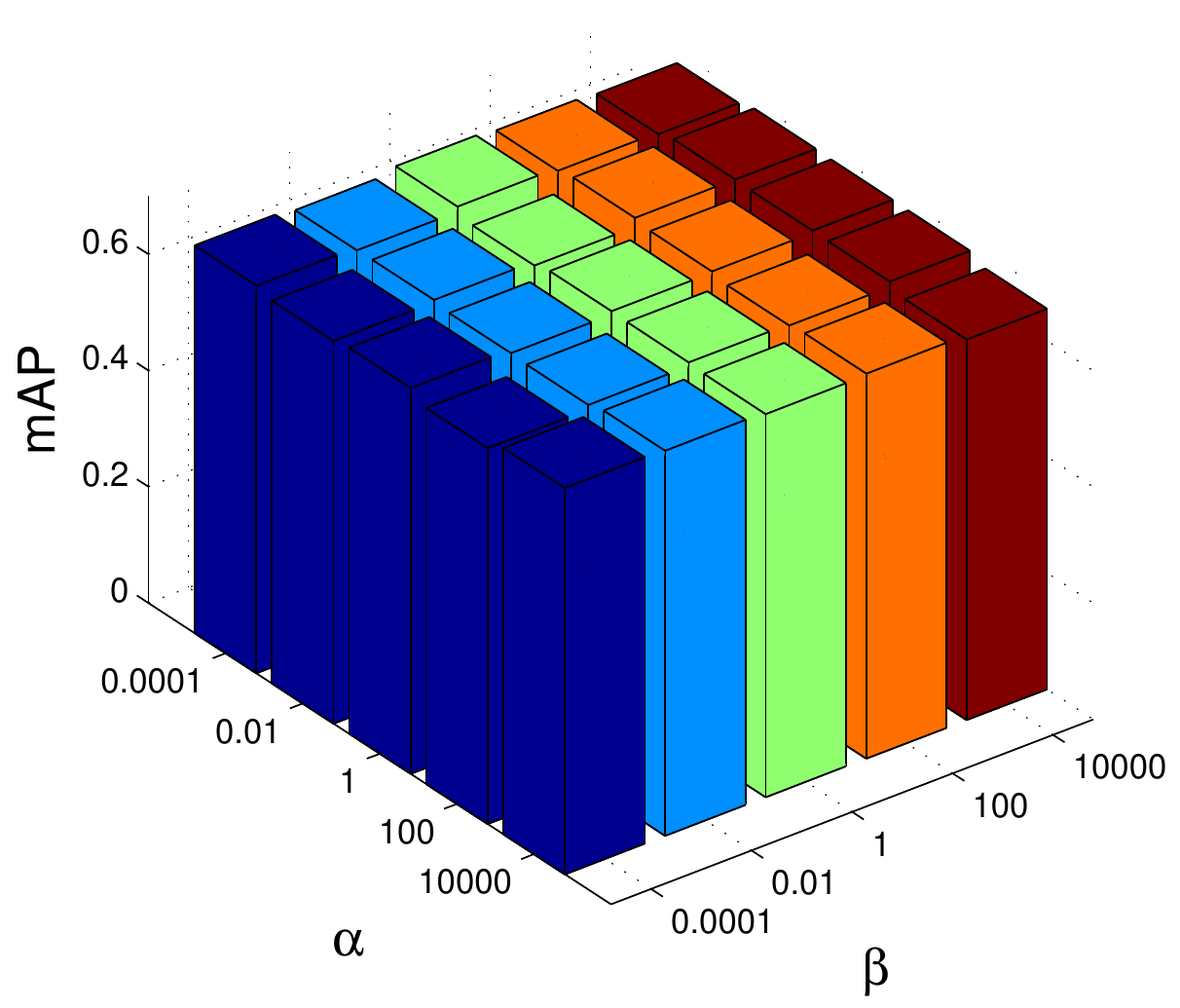}}
}\mbox{
\subfigure[]{\includegraphics[width=38mm]{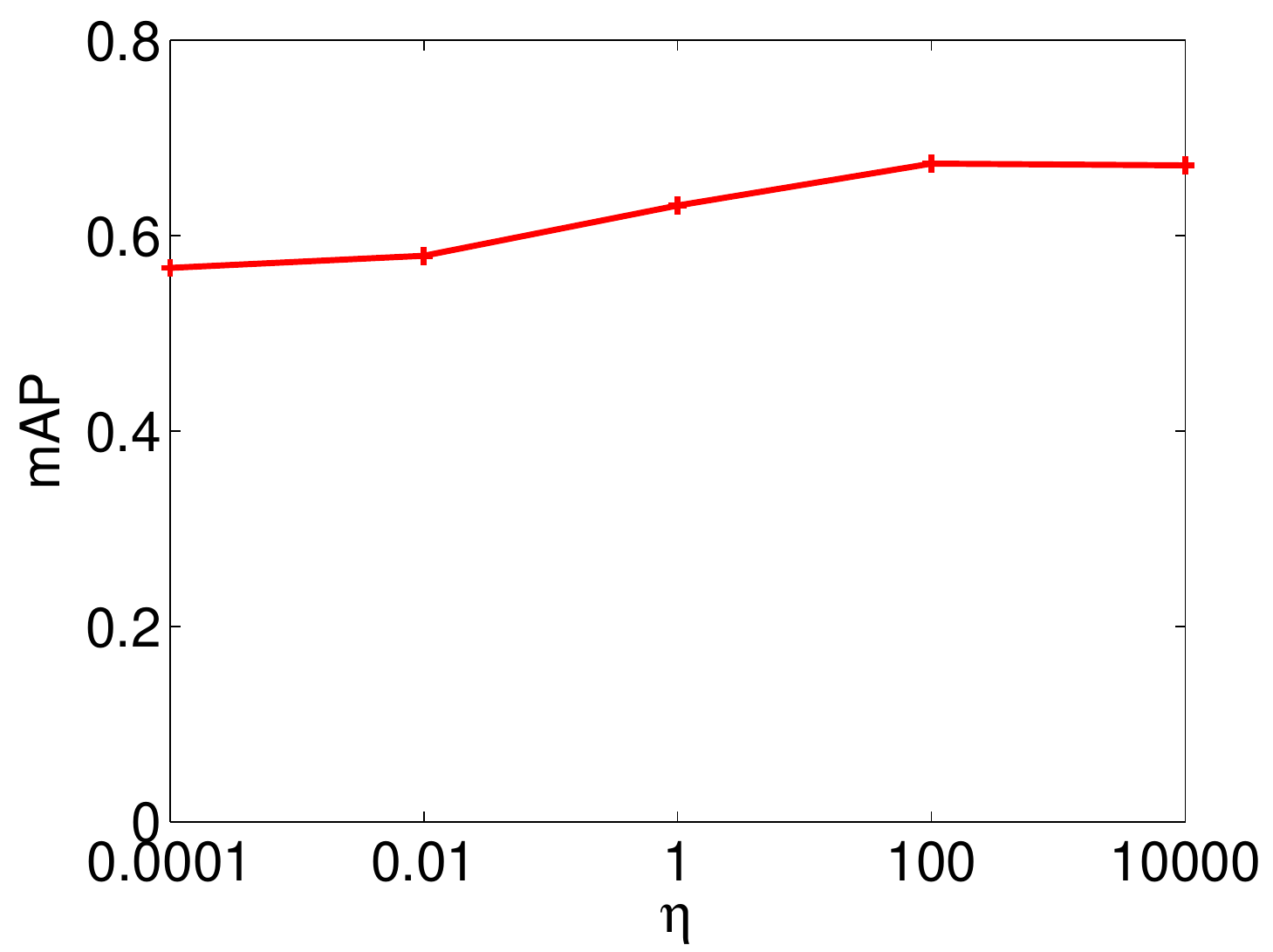}}
}
\caption{Performance variations with parameters on \textbf{MIR Flickr}.}
\label{mirflickr_para}
\end{figure*}
\begin{figure*}
\centering
\mbox{
\subfigure[$\alpha$ is fixed to 0.0001]{\includegraphics[width=38mm]{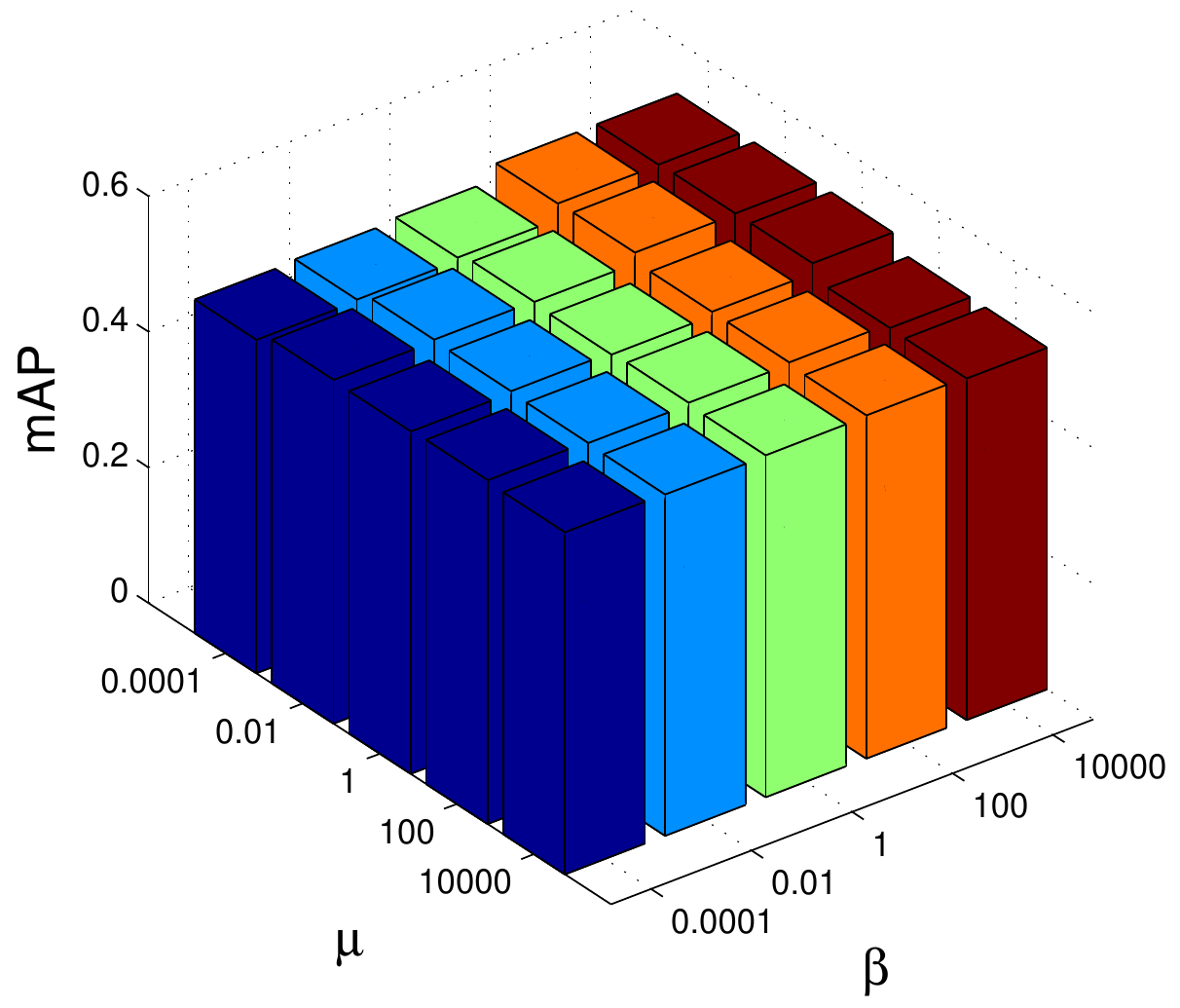}}
}\mbox{
\subfigure[$\beta$ is fixed to 100]{\includegraphics[width=38mm]{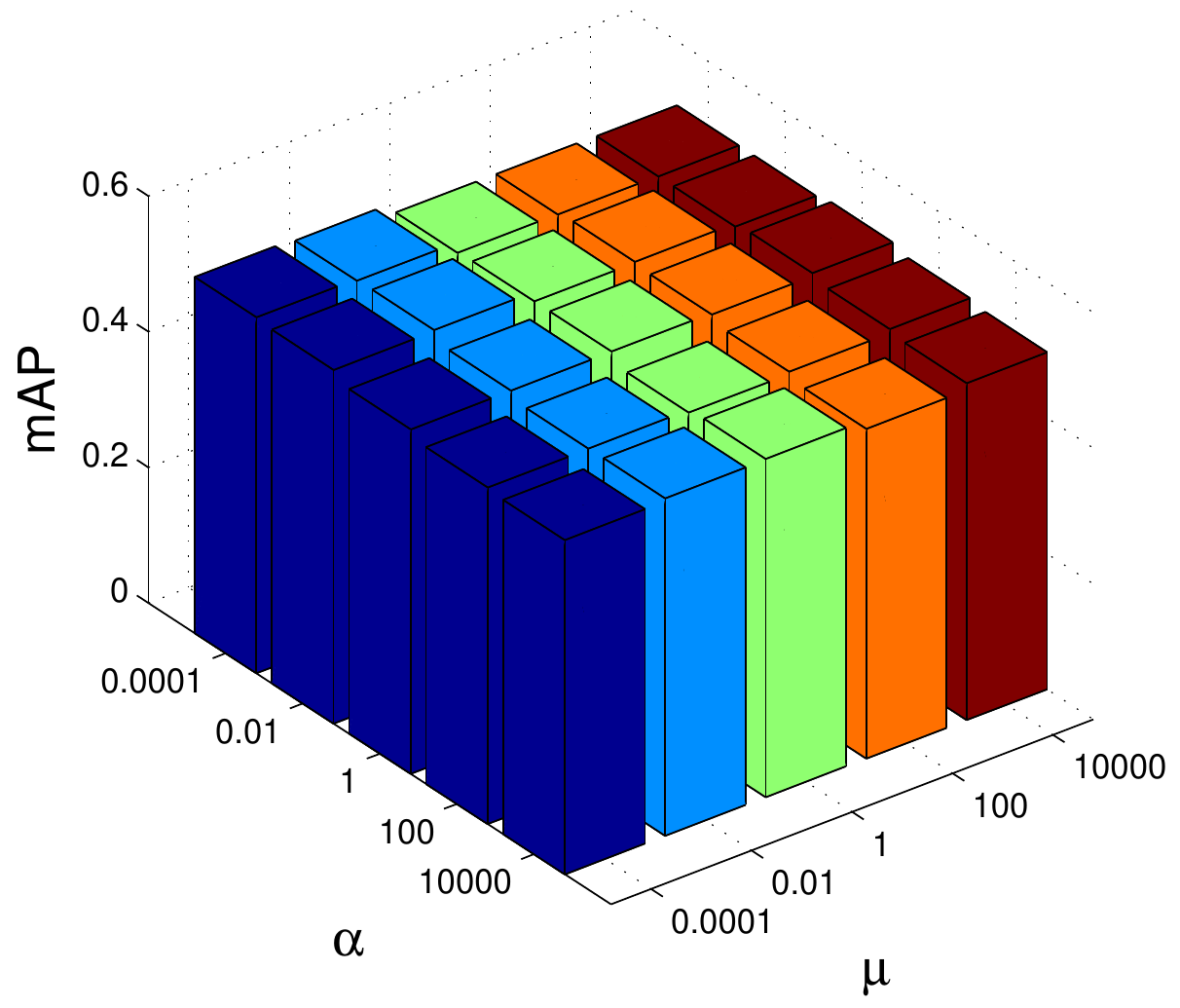}}
}\mbox{
\subfigure[$\mu$ is fixed to 0.0001]{\includegraphics[width=38mm]{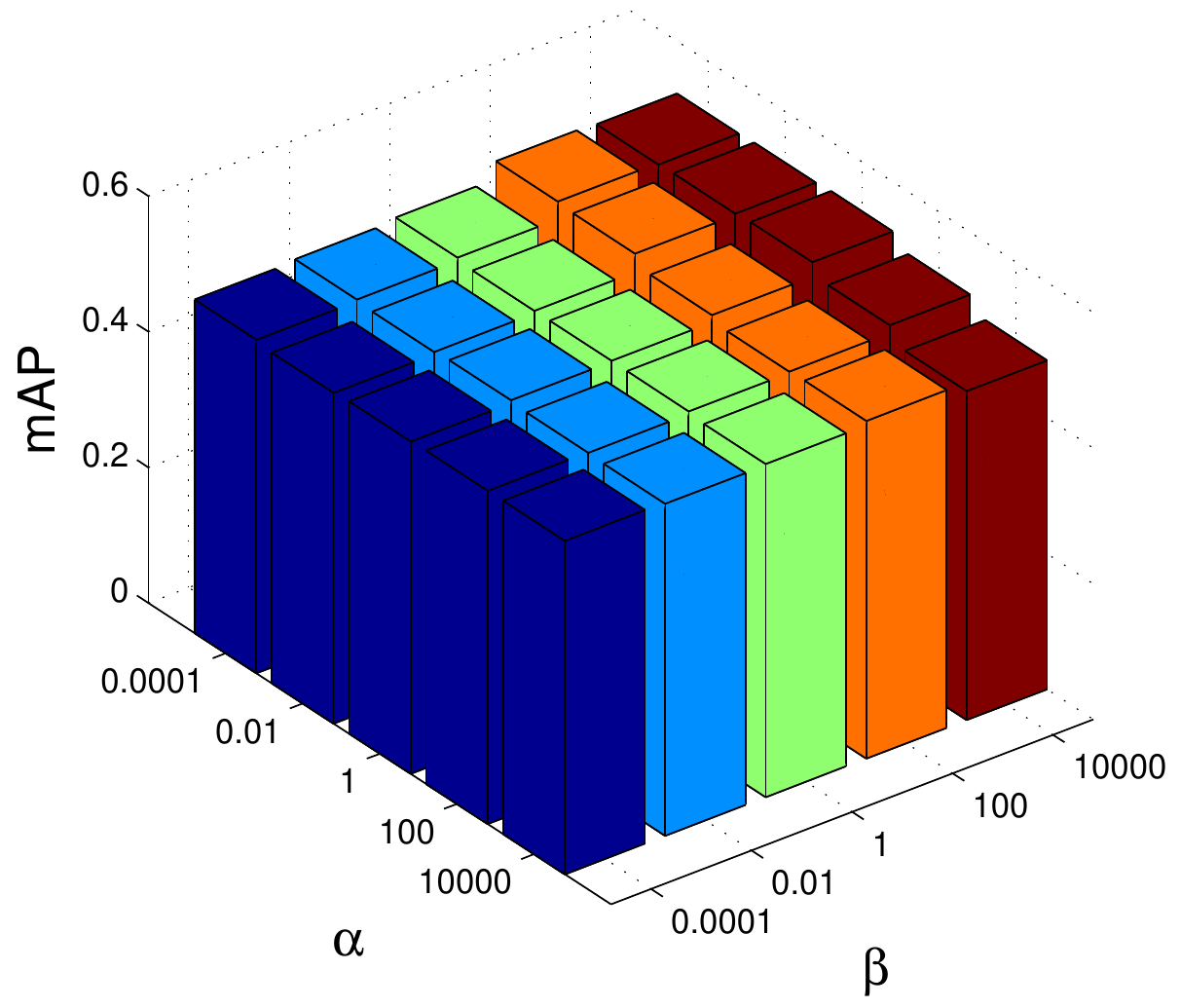}}
}\mbox{
\subfigure[]{\includegraphics[width=38mm]{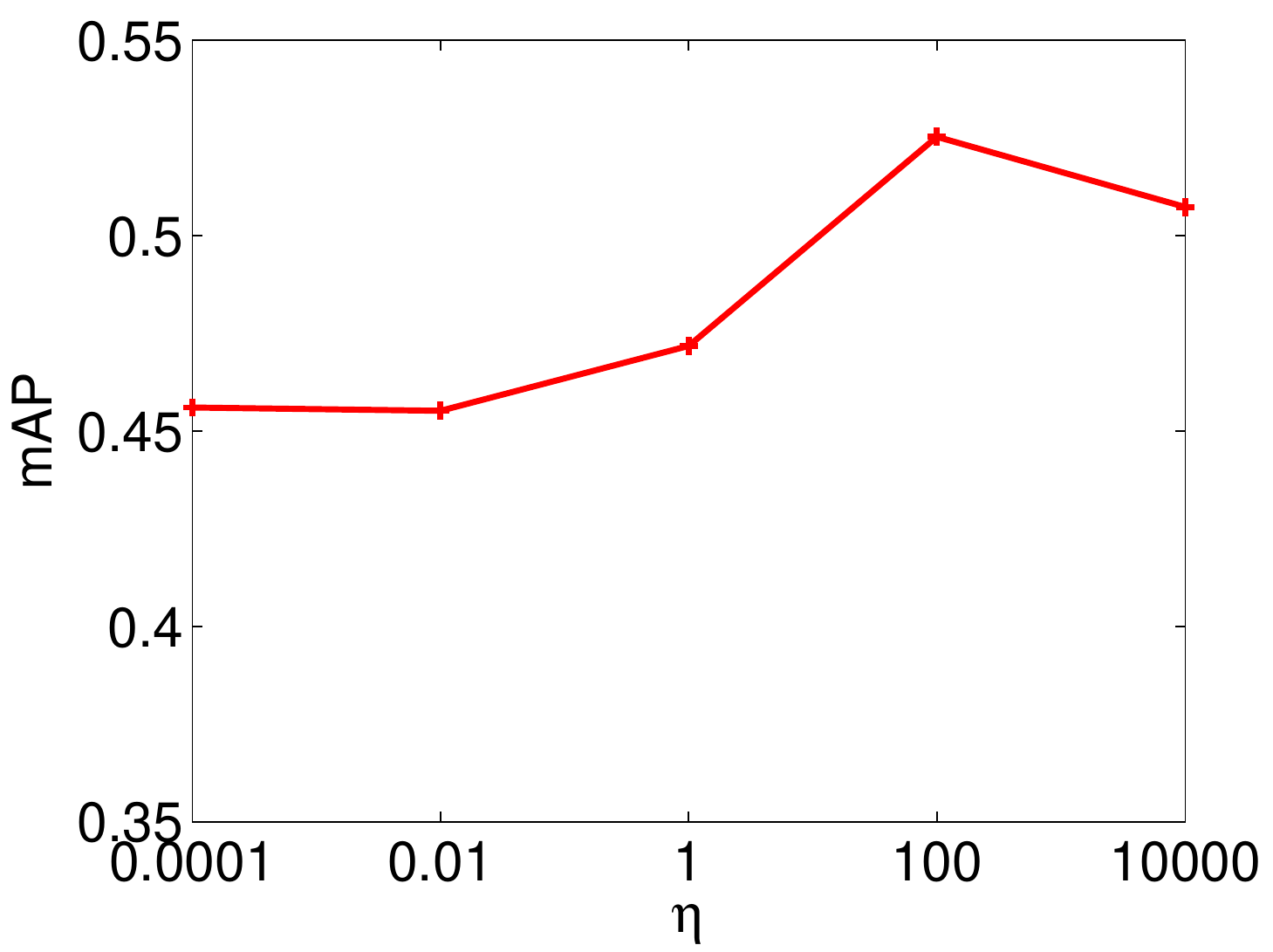}}
}
\caption{Performance variations with parameters on \textbf{NUS-WIDE}.}
\label{nuswide_para}
\end{figure*}

\subsection{Effects of Discrete Optimization}
Our approach can directly deal with the discrete constraint, bit-uncorrelation constraint and bit-balance constraint imposed on hash codes. To evaluate the effects of three constraints, we report the performance of DSTH respectively by relaxing the discrete constraint, removing bit-balance constraint and bit-uncorrelation constraint in the Eq.(\ref{eq:objective}). We denote DSTH-I as the approach that relaxes discrete constraint. In this experiment, we adopt conventional relaxing+rounding optimization in many existing hashing approaches to solve the hash codes. Specifically, the relaxed hashing values are first solved with ALM, but the final binary hash codes are generated by mean thresholding. We also compare the performance of DSTH with the variant approach DSTH-II that removes bit-balance constraint, and the variant approach DSTH-III that removes bit-−uncorrelation constraint. Table \ref{discretop} summarizes the comparison results. We can clearly observe that DSTH achieves superior performance in almost all cases. These results validate the effects of discrete optimization on direct semantic transfer and alleviating information loss. All three constraints contribute positively to the retrieval performance.
\begin{figure*}
\centering
\mbox{
\subfigure[\textbf{Wiki}]{\includegraphics[width=51.4mm]{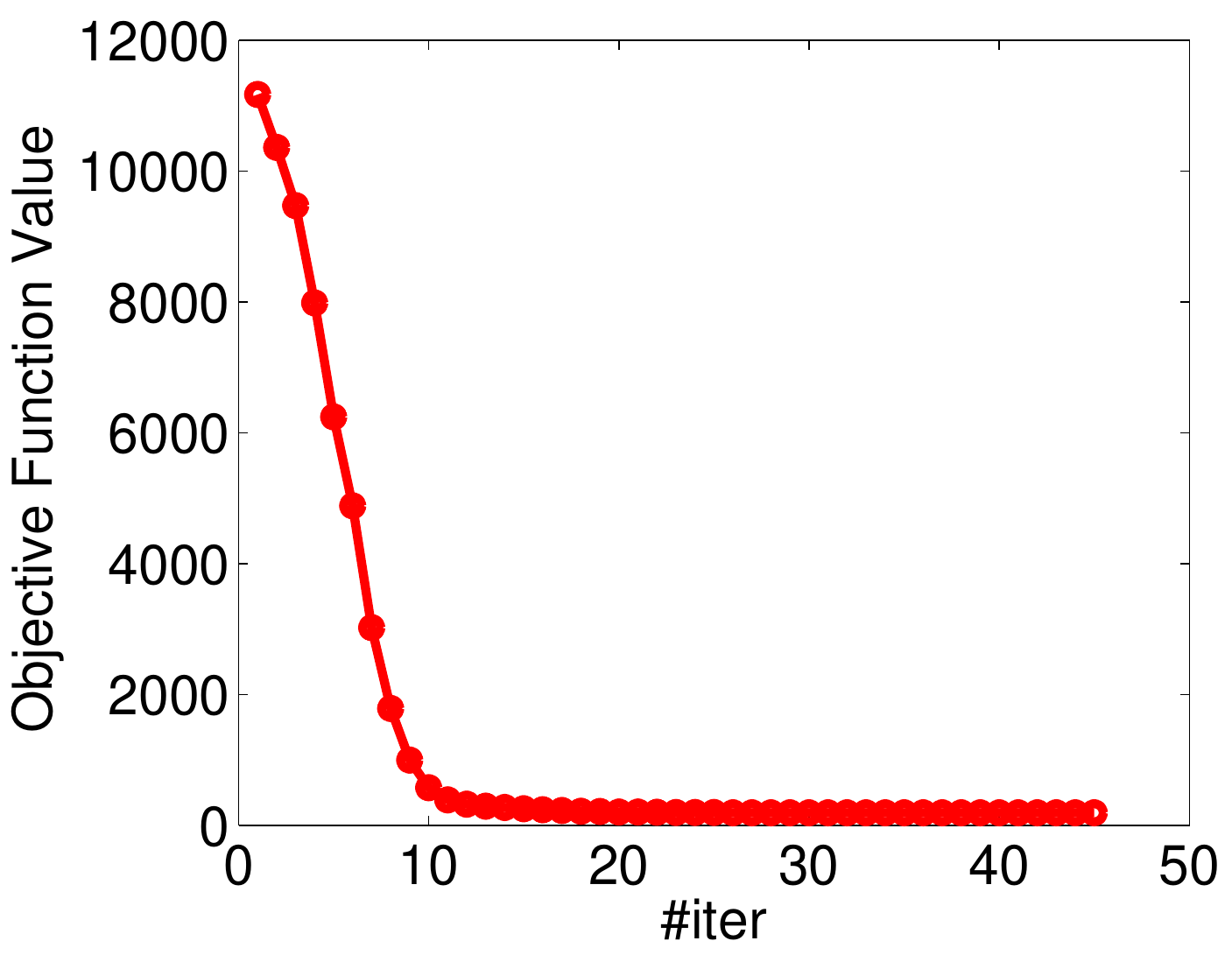}}
}\mbox{
\subfigure[\textbf{MIR Flickr}]{\includegraphics[width=51.4mm]{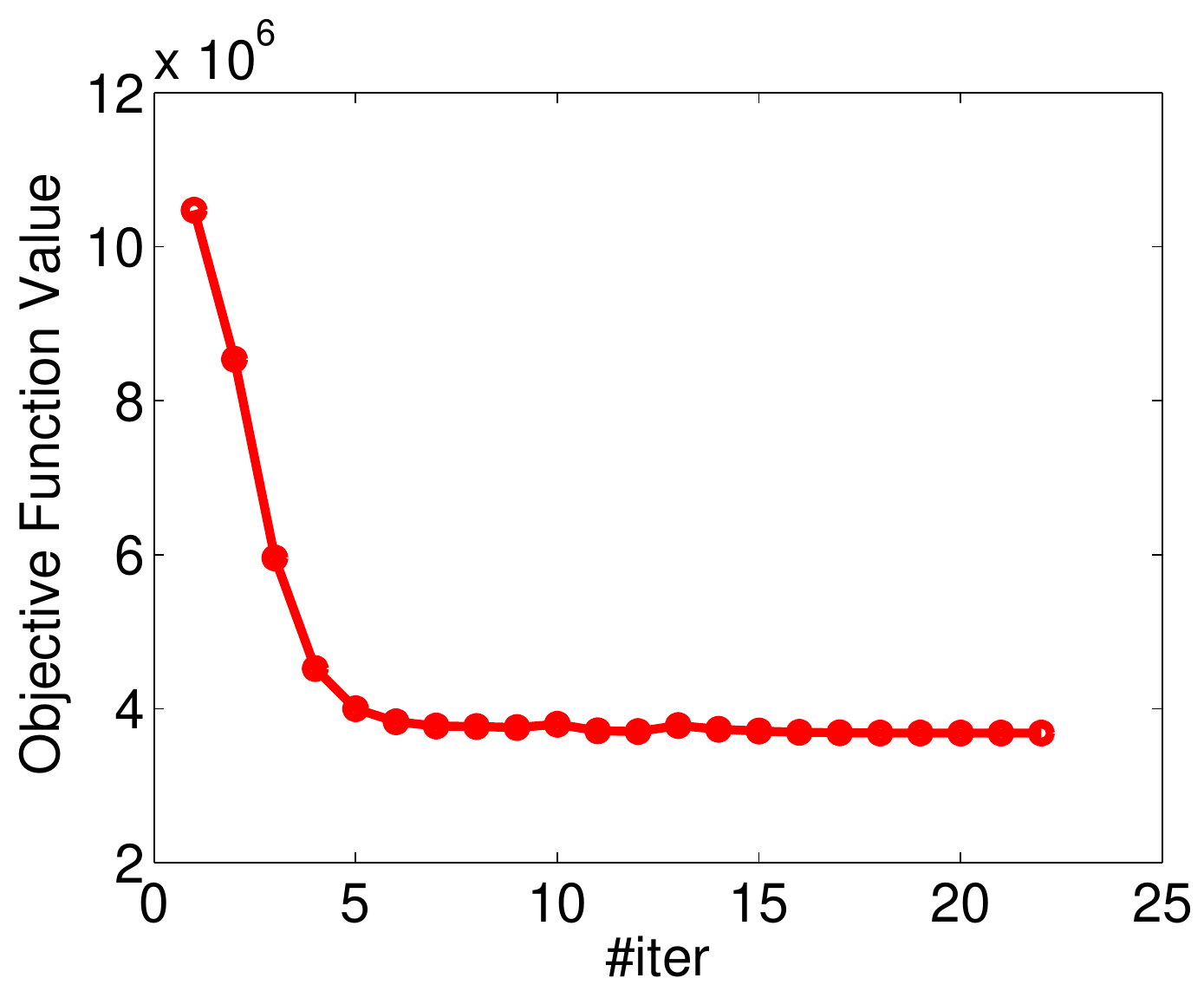}}
}\mbox{
\subfigure[\textbf{NUS-WIDE}]{\includegraphics[width=50mm]{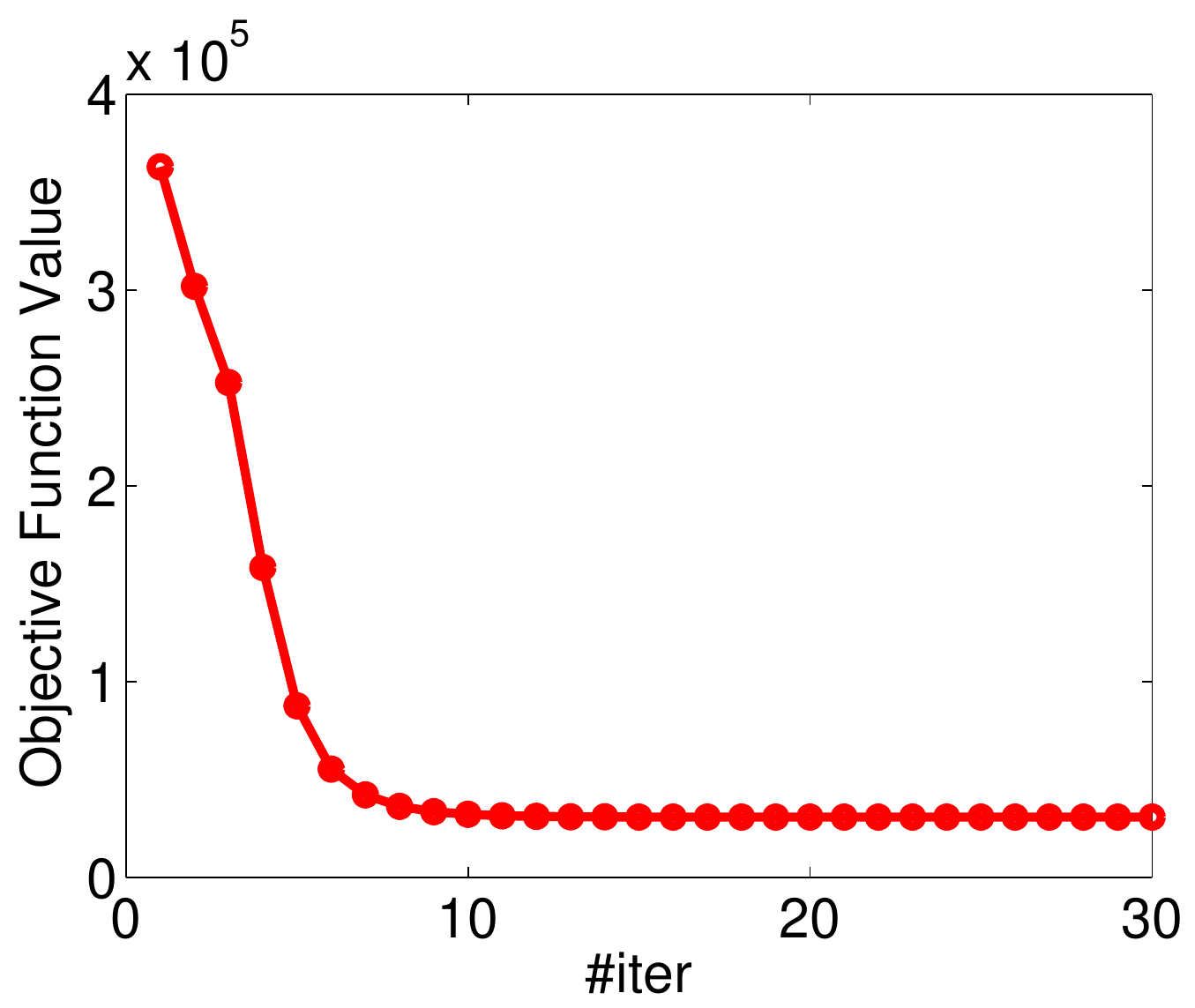}}
}
\caption{Objective function value variations with the number of iterations on three datasets.}
\label{nuswide_conv}
\end{figure*}
\subsection{Effects of Semantic Transfer}
Our approach explicitly exploits the latent semantics involved in contextual modalities to enhance the discriminative capability of discrete hash codes. In this subsection, we conduct experiment to investigate the effects of semantic transfer on the overall performance of DSTH. To this end, we compare the performance of DSTH with that of the approach variant which only considers visual similarity preservation (the first part of DSTH). We denote this variant as DSTH-IV. Table \ref{discretop} presents the comparison results. It shows that DSTH consistently outperforms the competitor on all code lengths and datasets. On \textbf{MIR Flickr} and \textbf{NUS-WIDE}, the largest performance increase can reach about 4\% and 7\% respectively. The performance increase is mainly attributed to the effective semantic enrichment of hash codes by semantic transfer. It also validates the fact that semantics in contextual modalities are indeed complementary with visual contents of images. In addition, we find from the table that the performance gap is different on different code lengths and datasets. This is attributed to the different effects of contextual modalities on enriching semantics of hash codes.

\subsection{Performance Variations with Training Size}
In this experiment, we evaluate the impact of training size on DSTH performance. We fix the hash code length to 128 and report the performance on \textbf{NUS-WIDE}. Table \ref{trainingsize} illustrates the performance variations with training size. We can easily observe that the performance of DSTH first increases with training data size and then becomes stable after certain point (training size 4.5K). Specifically, the gap between the performance obtained on 0.5K and that on 4.5K is 0.0235. DSTH can achieve satisfactory performance with a reasonably small training set. This experimental phenomenon illustrates that, even with small training data, DSTH can already effectively capture the valuable semantics to enhance the discriminative capability of image hash codes. It also validates the well training efficiency of DSTH when obtaining promising retrieval performance. In addition, it is interesting to find that, even with smaller training data, DSTH can achieve better performance than several compared approaches trained with more data. The reason of performance improvement is that the discovered semantics in contextual modalities can effectively mitigate the semantic shortage of shorter hash codes. These results also validate the effects of semantic transfer on enhancing the representation capability of hash codes.

In addition, we report the training time variations with training size. Figure \ref{fig:trainsize} illustrates the main results. We can easily observe that the training time increases linearly with training size. It validates the linear scalability of DSTH and demonstrates that it is suitable for large-scale datasets.

\subsection{Convergence Analysis}
As analysis in Section \ref{sec:3}, at each iteration, the updating of variables will monotonically decreases towards the lower-bounded objective function in Eq.(\ref{eq:objective}). Theoretically, the iterations will make the proposed discrete optimization method converge. In this subsection, we conduct experiment on \textbf{NUS-WIDE} with fixed hash code length 128 to verify this claim. Similar results can be obtained on other datasets and hash code lengths. Figure \ref{nuswide_conv} illustrates the main experimental results. We can clearly find that, on three datasets, the objective function value decreases sharply first and does not change significantly after several iterations (about 10). This result empirically validates that the convergence of DSTH can be achieved with augmented Lagrangian multiplier approach.

\subsection{Parameter Sensitivity Experiment}
In this subsection, we conduct empirical experiments to validate the parameter sensitivity of DSTH. More specifically, we observe the performance variations of DSTH with $\alpha$, $\beta$, $\mu$ and $\eta$. $\alpha$, $\beta$, and $\mu$ are used in discrete optimization (Eq.(\ref{eq:tj})), $\eta$ is used in hash function learning. They are all designed to play the trade-off between regularization terms and empirical loss. In experiment, we report the experimental results when these parameters are varied from $\{10^{-4}, 10^{-2}, 1, 10^2, 10^4\}$.  For $\alpha$, $\beta$, and $\mu$, as they are equipped in the same equation, we observe the performance variations with respect to two parameters while fixing the remaining one parameter. For $\eta$, we observe the performance variations by fixing $\alpha$, $\beta$, and $\mu$. Figure \ref{wiki_para}, \ref{mirflickr_para}, and \ref{nuswide_para} demonstrate the main experimental results. From these figures, we can find that the performance is relatively stable to a wide range of parameter variations ($\alpha$, $\beta$, $\mu$). And it can achieve the best result when $\eta$ is set to a certain value. The best performance is achieved when parameters are set as: \textbf{Wiki}$\{\alpha=0.0001, \beta=10000, \mu=1, \eta=1\}$, \textbf{MIR Flickr}$\{\alpha=100, \beta=10000, \mu=0.01, \eta=100\}$, \textbf{NUS-WIDE}$\{\alpha=0.0001, \beta=100, \mu=0.0001, \eta=100\}$.

\section{Conclusions and Future Work}
\label{sec:5}
Because of the intrinsic limitation of image representation on characterizing high-level semantics, existing hashing methods for scalable content-based image retrieval inevitably suffer from semantic shortage. In this paper, we propose the \emph{Discrete Semantic Transfer Hashing} (DSTH) to tackle the problem. It \emph{directly} exploits abundant auxiliary contextual modalities to augment the semantics of discrete image hash codes.  We formulate a unified hashing framework to simultaneously preserve visual similarities and perform semantic transfer. Moreover, to guarantee direct semantic transfer and avoid information loss, we explicitly impose the discrete constraint, bit-uncorrelation constraint and bit-balance constraint on hash codes. A novel and effective discrete optimization method with favorable convergence is developed to iteratively solve the optimization problem. The discrete hashing optimization has linear computation complexity and desirable scalability. Experiments on three benchmarks demonstrate the superior performance of DSTH compared with several state-of-the-art hashing methods.

In the future, inspired by the recent success of unsupervised deep hashing \cite{deepbit}, our work will be extended to learn a non-linear deep neural network based image hash function while resorting to the semantic augment from contextual modalities.

\ifCLASSOPTIONcaptionsoff
  \newpage
\fi

\section*{Acknowledgment}
Heng Tao Shen is corresponding author.
The authors would like to thank the anonymous reviewers for their constructive and helpful suggestions.

\ifCLASSOPTIONcaptionsoff
  \newpage
\fi

\bibliographystyle{IEEEtran}
\bibliography{IEEEabrv,IEEETNN}

\begin{IEEEbiography}[{\includegraphics[width=1in,height=1.25in,clip,keepaspectratio]{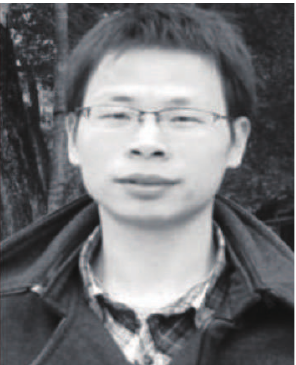}}]{Lei Zhu}
received the B.S. degree (2009) at Wuhan University of Technology, the Ph.D. degree (2015) at Huazhong University of Science and Technology. He is currently a full Professor with
the School of Information Science and Engineering, Shandong Normal University, China.  He was a Research Fellow under the
supervision of Prof. Heng Tao Shen at the University of Queensland (2016-2017), and Dr. Jialie Shen at the Singapore Management University (2015-2016).
His research interests are in the area of large-scale multimedia content analysis and retrieval.
\end{IEEEbiography}

\begin{IEEEbiography}[{\includegraphics[width=1in,height=1.25in,clip,keepaspectratio]{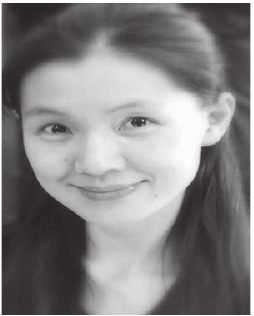}}]{Zi Huang}
received the B.Sc. degree from Tsinghua University, Beijing, China, in 2001, and the Ph.D.
degree in computer science from the University of Queensland, Brisbane, QLD, Australia, in 2004.
She is a Senior Lecturer and ARC Future Fellow with the School of Information Technology and Electrical Engineering,
University of Queensland. Her research interests
include multimedia search, social media analysis,
database, and information retrieval. She has authored
or coauthored papers that have been published
in leading conferences and journals, including ACM
Multimedia, ACM SIGMOD, IEEE ICDE, the IEEE Transactions ON Multimedia,
the IEEE Transactions on Knowledge and Data Engineering,
the ACM Transactions on Information Systems, and ACM Computing Surveys.
\end{IEEEbiography}
%

\begin{IEEEbiography}[{\includegraphics[width=1in,height=1.25in,clip,keepaspectratio]{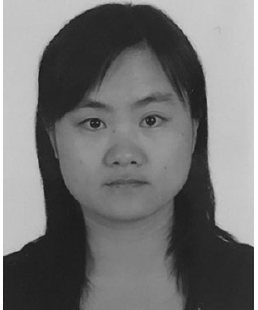}}]{Zhihui Li}
received the B.S. degree from Beijing University of Posts and Telecommunications in 2008. She is currently working as a research assistant at the School of Computer Science and Technology in Shandong University. After her graduation, she
has worked as a Data Analyst in Beijing Etrol Technologies Co., Ltd until December 2017. Her research interests include artificial intelligence, machine learning, and computer vision.
\end{IEEEbiography}

\begin{IEEEbiography}[{\includegraphics[width=1in,height=1.25in,clip,keepaspectratio]{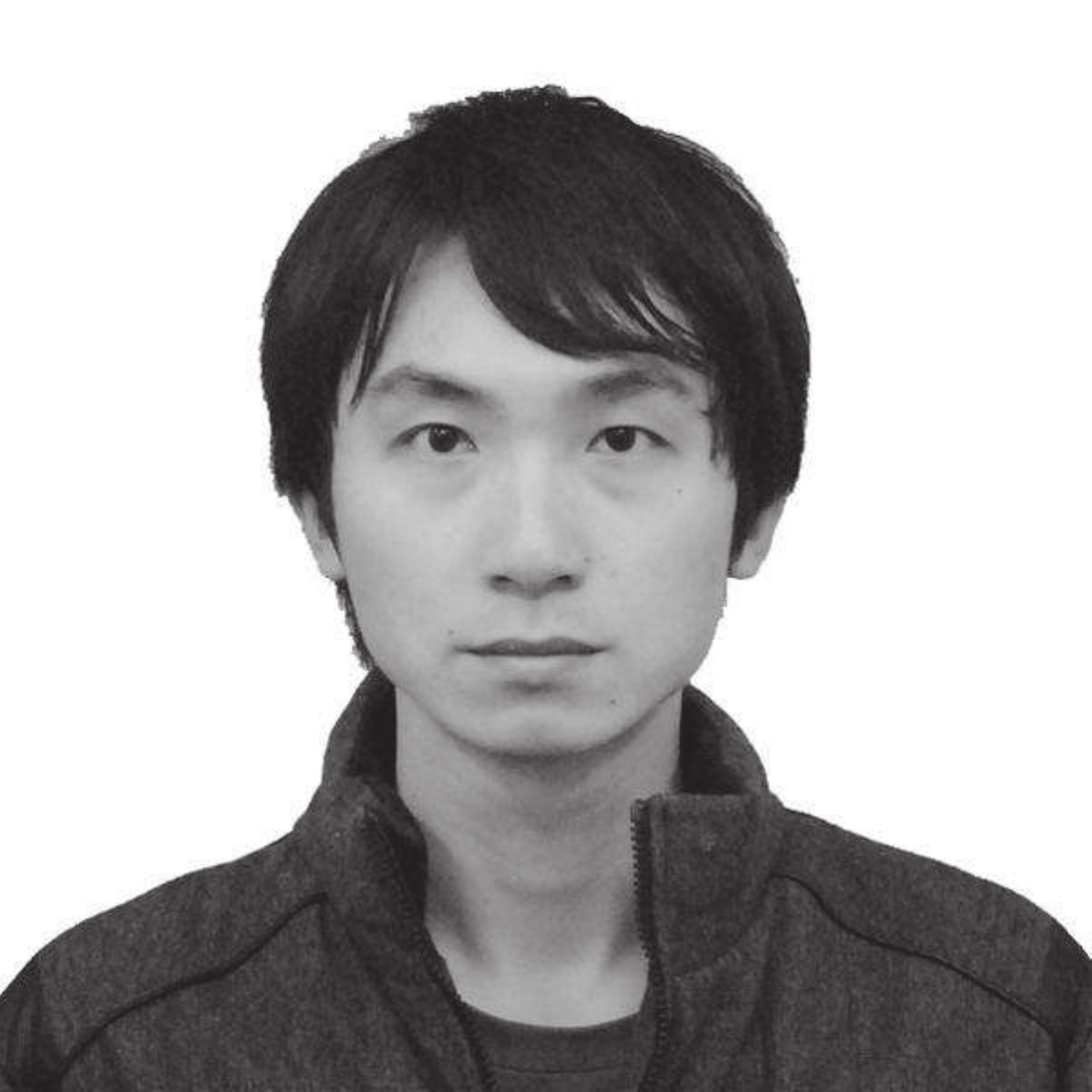}}]{Liang Xie}
received the B.S. degree from Wuhan University of Technology, China, in 2009, the Ph.D. degree from Huazhong University of Science and Technology, China, in 2015. He is currently an lecturer in the School of Science at Wuhan University of Technology. His current research interests include image semantic learning, cross-modal and multi-modal multimedia retrieval.
\end{IEEEbiography}
\begin{IEEEbiography}[{\includegraphics[width=1in,height=1.25in,clip,keepaspectratio]{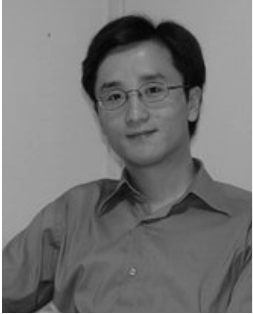}}]{Heng Tao Shen}
is currently a Professor of National "Thousand Talents Plan", the Dean of School of Computer Science and Engineering, and the Director of Center for Future Media at the University of Electronic Science and Technology of China. He is also an Honorary Professor at the University of Queensland. He obtained his BSc with 1st class Honours and PhD from Department of Computer Science, National University of Singapore in 2000 and 2004 respectively. He then joined the University of Queensland as a Lecturer, Senior Lecturer, Reader, and became a Professor in late 2011. His research interests mainly include Multimedia Search, Computer Vision, Artificial Intelligence, and Big Data Management. He has published 200+ peer-reviewed papers, most of which appeared in top ranked publication venues, such as ACM Multimedia, CVPR, ICCV, AAAI, IJCAI, SIGMOD, VLDB, ICDE, TOIS, TIP, TPAMI, TKDE, VLDB Journal, etc. He has received 6 Best Paper Awards from international conferences, including the Best Paper Award from ACM Multimedia 2017 and Best Paper Award - Honorable Mention from ACM SIGIR 2017. He has served as a PC Co-Chair for ACM Multimedia 2015 and currently is an Associate Editor of IEEE Transactions on Knowledge and Data Engineering.
\end{IEEEbiography}
\vfill
\end{document}